\documentclass[12pt]{article}
\usepackage[a4paper, total={6.5in, 10in}]{geometry}
\usepackage{amsmath}
\usepackage{graphicx}
\usepackage{enumerate}
\usepackage{natbib}
\usepackage{url} % not crucial - just used below for the URL 

\usepackage{xcolor,colortbl}
\usepackage{float}

\definecolor{Gray}{gray}{0.85}
\definecolor{LightCyan}{rgb}{0.88,1,1}

\usepackage{psfrag,epsf}
\usepackage{enumerate}

\usepackage{url} % not crucial - just used below for the URL 

\usepackage{wrapfig}
\usepackage{subcaption}
\usepackage{placeins} % for float barrier

\usepackage{graphicx}
\usepackage{tikz}
\usetikzlibrary{fit,positioning,arrows,automata}
\usetikzlibrary{shapes,shadows,arrows,positioning,graphs}

\usepackage{bigints}

\usepackage{amssymb,verbatim,color}
\usepackage{bm,lscape}
\usepackage{mathrsfs}
\usepackage{mathtools}
\usepackage{wasysym}
\usepackage{bbm}
\RequirePackage[colorlinks,citecolor=blue,urlcolor=blue]{hyperref}
\usepackage{subcaption}
\usepackage{algorithm}
\usepackage[noend]{algpseudocode}

\definecolor{darkblue}{rgb}{0.0, 0.0, 0.55}

\newcommand{\dd}{\mathrm{d}}

\newcommand{\beq}{\begin{equation}}
	\newcommand{\eeq}{\end{equation}}

\usepackage{amsthm}
\theoremstyle{plain}
\newtheorem{thm}{\textbf{Theorem}}
\newtheorem{defi}{\textbf{Definition}}
\newtheorem{lem}{\textbf{Lemma}}
\newtheorem{cor}{\textbf{Corollary}}

\newtheorem{prp}{\textbf{Proposition}}

\usepackage{fix-cm}    
\makeatletter
\newcommand\HUGE{\@setfontsize\Huge{70}{100}}
\makeatother

\begin{document}

%\bibliographystyle{natbib}

%%%%%%%%%%%%%%%%%%%%%%%%%%%%%%%%%%%%%%%%%%%%%%%%%%%%%%%%%%%%%%%%%%%%%%%%%%%%%%
\title{\bf Conditional partial exchangeability: \\
  a probabilistic framework for multi-view clustering}
\author{ Beatrice Franzolini\thanks{Franzolini was partially supported by the National Recovery and Resilience Plan of Italy (PE1 FAIR—CUP B43C22000800006).}\\
    Bocconi Institute for Data Science and Analytics, Bocconi University\\
    and \\
    Maria De Iorio\thanks{This work was supported by the Singapore Ministry of Education Academic Research Fund Tier 2 under Grant MOE-T2EP40121-0004.} \,\,and Johan Eriksson \\
    Yong Loo Lin School of Medicine, National University of Singapore,\\ 
    Institute for Human Development and Potential, 
    A*STAR}
 \date{ }
  \maketitle

\bigskip
\begin{abstract}
Standard clustering techniques assume a common configuration for all features in a dataset. However, when dealing with multi-view or longitudinal data, the clusters' number, frequencies, and shapes may need to vary across features to accurately capture dependence structures and heterogeneity. In this setting, classical model-based clustering fails to account for within-subject dependence across domains. We introduce conditional partial exchangeability, a novel probabilistic paradigm for dependent random partitions of the same objects across distinct domains. Additionally, we study a wide class of Bayesian clustering models based on conditional partial exchangeability, which allows for flexible dependent clustering of individuals across features, capturing the specific contribution of each feature and the within-subject dependence, while ensuring computational feasibility.
\end{abstract}

\noindent%
\vfill

\newpage
\section{Clustering multi-view information} 
Clustering is arguably the most famous unsupervised learning technique. It involves grouping observations into clusters based on their \textit{similarities}. 
Standard clustering techniques assume a common clustering configuration of subjects across all features observed in a sample. However, given the complexity and dimension of modern datasets,  a unique clustering arrangement for all the features is often inadequate to describe the structure and the heterogeneity in the population under study.  For instance, in longitudinal data analysis, the underlying clustering structure of individuals is likely to change over time; in multi-view data \citep[see, for instance,][]{yang2018multi} multivariate information collected across distinct domains may require clusters' shapes and definitions to change from feature to feature. In this setting, a unique clustering configuration based on all the observed features not only may be hard to detect \citep[often leading to clusters of small size to accommodate heterogeneity in multi-dimensional spaces, cf.,][]{chandra2020escaping}, but will also mainly capture global patterns, down-weighting the idiosyncratic contribution of each feature. Moreover, it depends on features' dimensions, favoring higher dimensional features as more important in explaining the heterogeneity across subjects. See Figure~\ref{fig:dimEX} for a toy example illustrating the latter problem.  
In this work, we focus on clustering problems where multi-view or longitudinal information is available for the same subjects, and we allow the underlying clustering structure to change across features/time.

\begin{figure}[tbh]
    \centering
    \begin{subfigure}[t]{0.32\textwidth}
        \includegraphics[width=\textwidth]{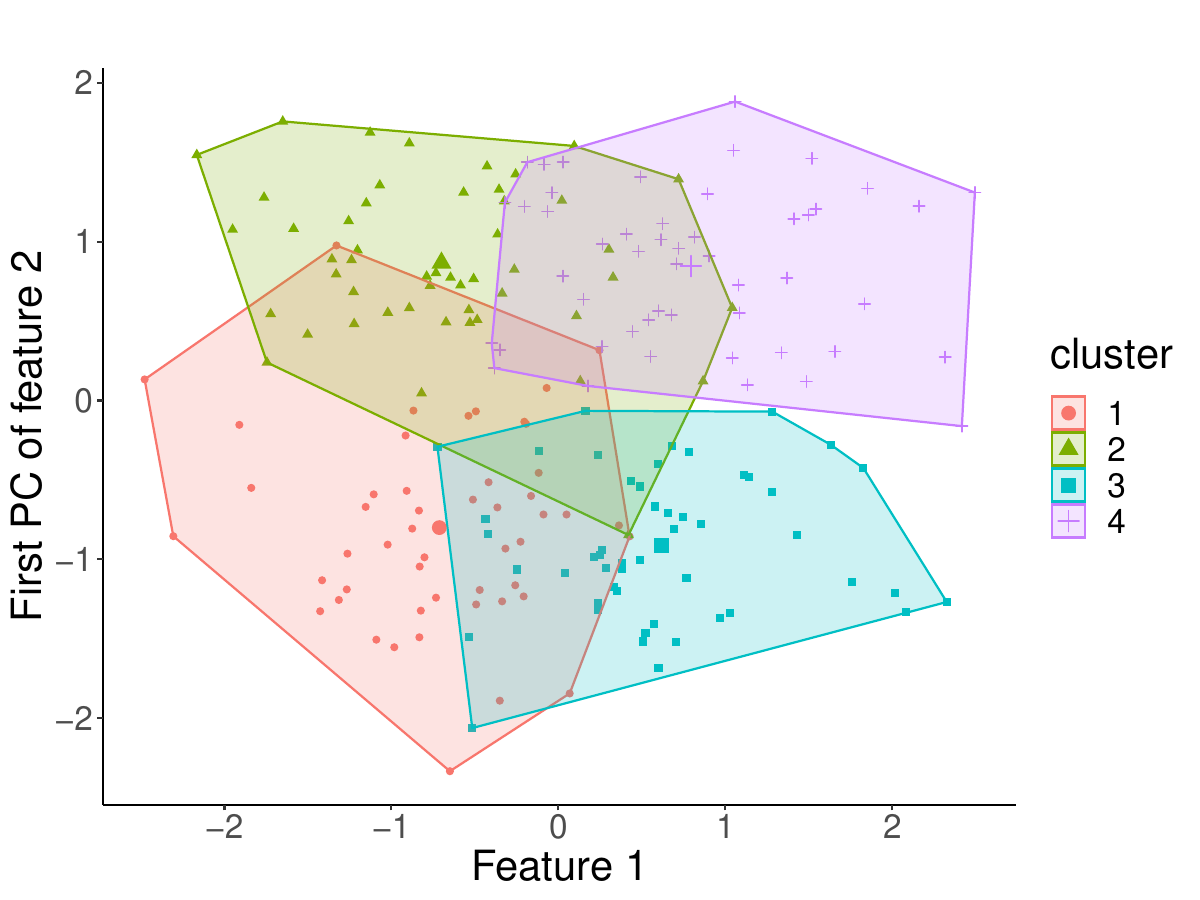}
        \caption{True clustering structure. Each point corresponds to a subject. Colors represent the true cluster assignment.}
    \end{subfigure}
    \hfill
    \begin{subfigure}[t]{0.32\textwidth}
        \includegraphics[width=\textwidth]{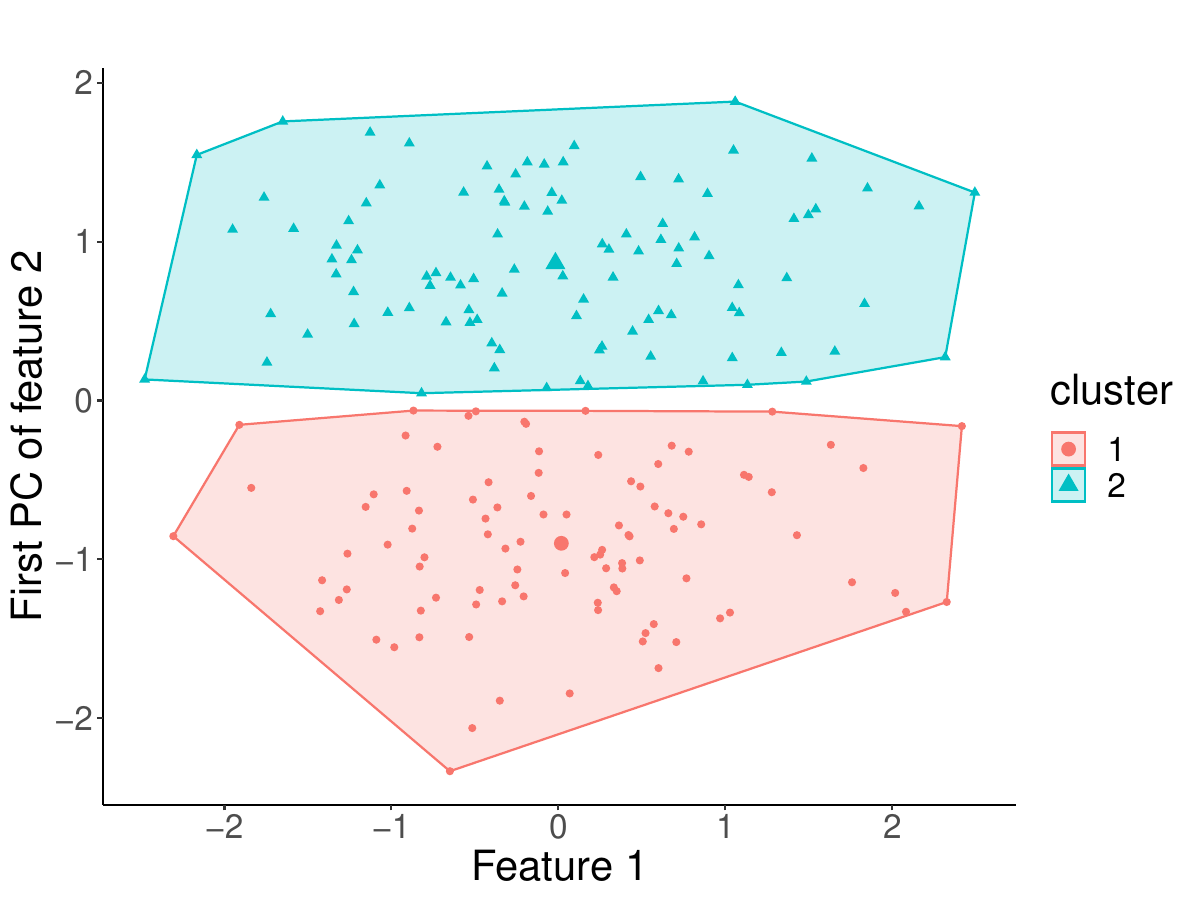}
        \caption{k-means clustering configuration with the number of clusters determined by elbow plot, gap statistics, and silhouette method.}
    \end{subfigure}
    \hfill
    \begin{subfigure}[t]{0.32\textwidth}
        \includegraphics[width=\textwidth]{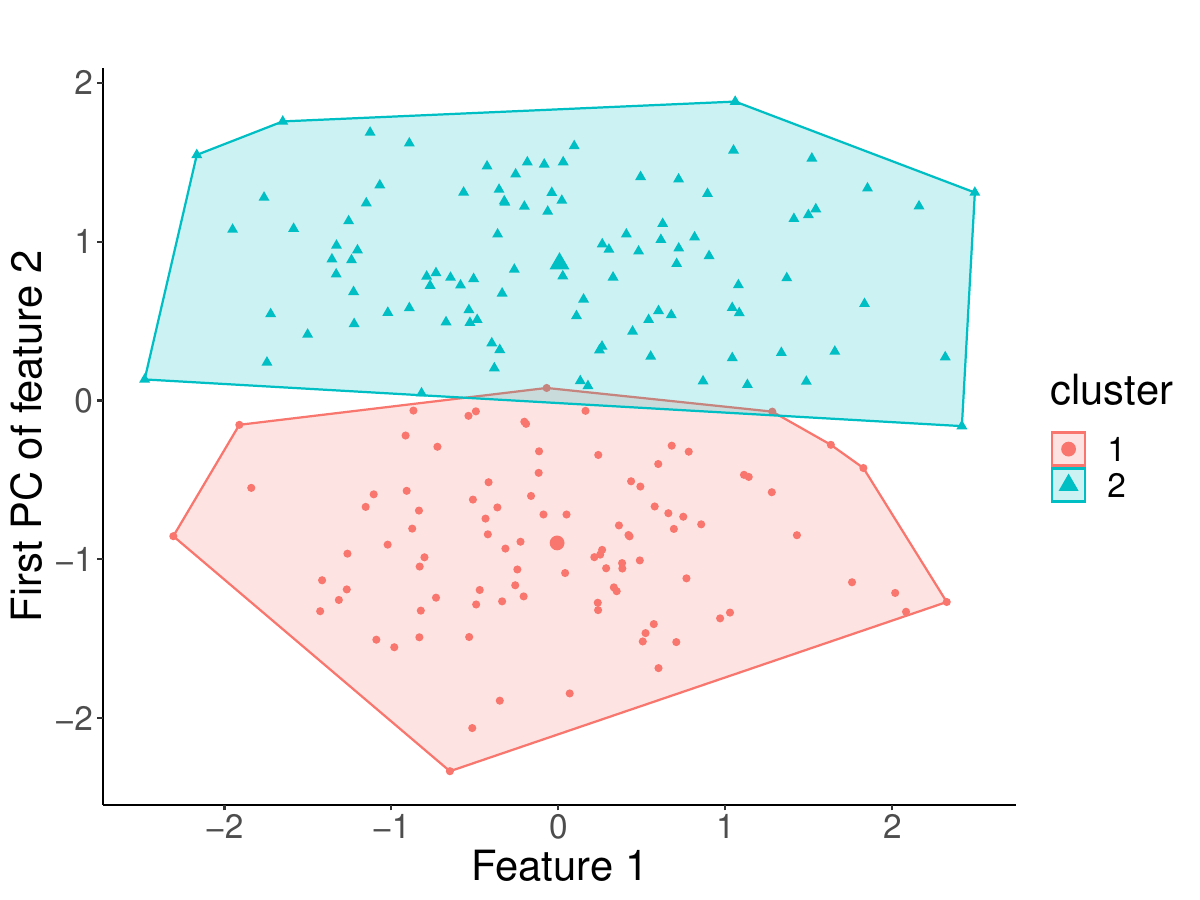}
        \caption{Dirichlet process mixture estimate of the clustering configuration obtained minimizing the variation of information loss function.}
    \end{subfigure}
    \caption{\label{fig:dimEX}Toy example. Data was simulated for two features and 200 subjects, with the true clustering configuration displayed in panel (a). The first feature (x-axis) is sampled from a univariate Normal with unitary variance and mean equal either to 1 or -1, depending on the true cluster assignment. The second feature (y-axis) is sampled from a three-variate Normal with identity covariance matrix and mean equal to either $(1,1,1)$ or $(-1,-1,-1)$ depending on the true cluster assignment; it is represented via its first principal component on the y-axis. The Dirichlet process estimate in panel c) is obtained with a Multivariate Normal kernel with identity covariance matrix, standard Multivariate Normal base measure, and concentration parameter equal to 0.1.  Both the clustering configurations obtained with k-means (panel b) and the Dirichlet process mixtures (panel c) are heavily informed by the second feature and ignore the information contained in the first feature.}
\end{figure}

The two main approaches for clustering are model-based and algorithmic methods. Model-based methods rely on distributional assumptions about the underlying data-generating mechanism of the observations in each cluster, leading to a mixture model. 
Unlike algorithmic techniques, model-based methods explicitly define the shape of a cluster in terms of probability distribution functions. Most popular model-based approaches include Bayesian infinite mixture models \citep{ferguson1983bayesian,lo1984class, barrios2013modeling} and Bayesian mixtures with a random number of components \citep{nobile1994bayesian,deblasi2015,miller2018mixture, argiento2019infinity}. They allow for data-driven automatic selection of the number of clusters for which no finite upper bound has to be fixed. 
Most traditional clustering approaches (both model-based and algorithmic) are designed for single-view data and aim at detecting a unique clustering configuration of individuals in a sample. In recent years, a wealth of proposals for algorithms to integrate multi-view information has appeared in the machine learning literature \citep[see,][for comprehensive reviews of the topic]{ yang2018multi,chen2022representation}. Nonetheless, such methods, while recognizing the multi-view nature of the data, provide again a single clustering configuration common to all the features, which may still fail to highlight the complementary information of each feature \citep{yao2019multi}. An interesting exception is provided by the algorithm proposed by \cite{yao2019multi}.
In the  Bayesian clustering literature, the focus is often placed on \emph{multi-sample data}, rather than \emph{multi-view} data, in the sense that there is an initial natural grouping of the subjects (for example, based on treatment groups in a clinical trial, or some level of a particular covariate) which is treated as deterministic and there is no overlap of subjects across groups. Then, clustering is performed within each group, with clusters possibly shared among groups. 
These models are obtained by inducing dependence between the group-specific random probability measures in the underlying mixture model and they have been the object of extensive research in recent decades \citep[see,][]{maceachern2000dependent, mueller2004amethod, teh2006teh, caron2007generalized, dunson2008kernel, ren2008dynamic, dunson2010nonparametric, rodriguez2010latent, taddy2010autoregressive, rodriguez2011nonparametric, lijoi2014bayesian, foti2015survey, caron2017generalized, griffin2017compound, deyoreo2018modeling, argiento2020hierarchical, bassetti2020hierarchical, ascolani2021predictive, beraha2021semi, denti2021common, zhou2021disentangled, quintana2022dependent, lijoi2023flexible}. Models built with this strategy may be effectively employed for clustering multi-sample data, i.e., when different clustering configurations refer to disjoint sets of subjects. 
However, we note that these are not suitable for multi-view data. 
As we show in this work, when they are applied to cluster multi-view or longitudinal data, such methods focus on marginal inference based on each feature and fail to capture the true nature of the multivariate dependence. In particular, in Section~\ref{s:mainCPE}, we show how this is a consequence of the fact that they disregard subjects' identities, i.e., that subjects are indeed the same for all the observed features or times. 

The Bayesian literature on clustering methods for longitudinal or multi-view information is rather limited. In this context,  the core challenge is to define a probabilistic model able to account for subjects' identities across multiple features. 
Bayesian clustering approaches that allow to both preserve subjects' identity and provide multiple clustering configurations, appear limited to the following specific models: the hybrid Dirichlet process \citep{petrone2009hybrid}, the enriched Dirichlet process \citep{wade2011enriched}, the separately exchangeable random partition models in \cite{lee2013nonparametric} and \cite{lin2021separate} and the temporal random partition model of \cite{page2022dependent}. Even though these models are quite different in nature, they all belong to the novel probabilistic framework we introduce here, which also enables gaining new insights about these existing approaches. 

The main contribution of this work is the introduction of \emph{conditional partial exchangeability} (CPE), a modeling principle for multi-view and dynamic probabilistic clustering. CPE is a condition imposed on the conditional law of the observable in multi-view data, inducing dependence across distinct clustering configurations of the same subjects while preserving their identities. An additional contribution is developing a specific class of mixture models that satisfy CPE, which we refer to as \emph{telescopic clustering models}. The introduction of these models further highlights the utility of CPE. We show that these models are analytically and computationally tractable and establish Kolmogorov consistency of the distribution of the observable. Finally, we investigate two models within this class in more detail: one with an infinite number of components and another with a random number of components. We provide algorithms for posterior estimation and showcase the performance of the proposed framework on real and simulated data.

%%%%%%%%%%%%%%%%%%%%%%%%%%%%%%%%%%%%%%%%%%%%%%%%%%%%%%%%%%%%%%%%%%%%%%%%%%%%%%%%%%%%%%%%%%%%%%%%%%%%%%%%%%%%%%%%%%%%%%%%%%%%%%%%%%%
\section{Conditional partial exchangeability }\label{s:mainCPE}
Let  $(X_{1i}, X_{2i})$, be features on   the $i-$th observational unit, with $i=1,\ldots,n$. For simplicity of explanation, we partition the feature vector into two sub-components and discuss how to extend to a number $L$ of components in Section~\ref{ss:Llayers}: $X_{1i} \in \mathbb{X}_1\subset\mathbb{R}^{{d}}$ is the observation recorded at layer 1, which can represent, for example,  either a vector of \emph{primary} features or observations corresponding to the initial time point $t=1$, and $X_{2i} \in \mathbb{X}_2 \subset \mathbb{R}^{p}$ is the observation recorded at layer 2, which can refer to either a vector of \emph{secondary} features or observations corresponding to a subsequent time point $t=2$.
Importantly, the support spaces $\mathbb{X}_1$ and $\mathbb{X}_2$ are not assumed to coincide. In particular, the dimensions $d$ and $p$ may be different.
Our goal is to estimate two clustering configurations $\rho_1$ and $\rho_2$, which correspond to the first and second layers, respectively, allowing for dependence between the two clustering configurations and employing a learning mechanism that takes into account subjects' identity thus capturing multivariate dependence. The partition  $\rho_j$, $j=1,2$, can be represented by the vector $\bm{c}_j= (c_{j1},\ldots,c_{jn})$ of subject-specific allocation variables, whose elements take value in the set $[n]:=\{1,2,\ldots,n\}$ and are such that $c_{ji} = c_{jl}$ if and only if subjects $i$ and $l$ belong to the same cluster at layer $j$. 
As typically done, we assume exchangeability for the joint observations $(X_{1i}, X_{2i})_{i\geq1}$, i.e., $\mathbb{P}\left[(X_{1i}, X_{2i})_{i=1}^n\in A\right] = \mathbb{P}\left[(X_{1\sigma(i)}, X_{2\sigma(i)})_{i=1}^n\in A\right]  $
for any $\sigma$  permutation of $n$ elements, $n\geq1$, and measurable set $A\subseteq (\mathbb{X}_1\times\mathbb{X}_2)^n$. 
A review of key preliminaries on exchangeable partitions and partial exchangeability is provided in Section S1 of the Supplement. For the sake of clarity, assume that the clustering configurations fully capture the  dependence structure between first and second layers features, i.e.,
$	(X_{11},\ldots,X_{1n})\perp (X_{21},\ldots,X_{2n}) \mid \rho_1, \rho_2$.
The latter is a common assumption in clustering models for multivariate data \citep[see, e.g.,][]{rogers2008investigating, kumar2011co, lock2013bayesian, gao2020clusterings, franzolini2022bayesian} as it often avoids identifiability issues. Nonetheless, it can be relaxed.

Exchangeability of the bivariate observations does not imply \emph{conditional exchangeability} of one layer given the other, which is instead undesirable in a multi-view setting. 
To better understand this point, let us see what happens if we do assume exchangeability for observations in the second layer conditionally on the first-layer partition, $\rho_1$. Formally, $(X_{2i})_{i\geq1}$ is conditionally exchangeable with respect to $\rho_1$, if, for any realization of $\rho_1$, for any $n\geq1$ and for any permutation $\sigma$,
\begin{equation}\label{eq:condexc}
p(X_{21},\ldots,X_{2n} \mid \rho_{1n}) = p(X_{2\sigma(1)},\ldots,X_{2\sigma(n)}\mid\rho_{1n})
\end{equation}
where $p$ denotes a joint density function (to be understood as a mass function in the discrete case).
This implies, for instance, that the joint distribution of a pair of second-layer observations is invariant to their clustering allocations at layer 1, in the sense that, for any set of three subjects $i,j$ and $k$, $p(X_{2i}, X_{2j} \mid c_{1i} = c_{1j} \neq c_{1k} ) = p(X_{2i}, X_{2k} \mid c_{1i} = c_{1j} \neq c_{1k})$. 
Thus, under conditional exchangeability, knowing that subjects $i$ and $j$ belong to the same cluster at layer 1 does not provide any information at layer 2 specific to those same two subjects, i.e., the model does not preserve subjects' identities. 
This unusual behavior of the learning mechanism under \emph{conditional exchangeability} arises because this assumption restricts the transfer of subject-level information (e.g., which subjects belong to the same cluster in $\rho_1$) between layers. Instead, only population-level information, such as the number of clusters and their frequencies in  
$\rho_1$, is propagated to the next layer. \citep[cf.,][]{page2022dependent}. 
 However, this assumption is also at the core of the majority of dependent Bayesian clustering methods \citep[see,][for a recent review]{ quintana2022dependent}, which are thus appropriate for multi-sample data but cannot be effectively applied in longitudinal and multi-view settings. 
An alternative to \emph{conditional exchangeability} (of the second layer given $\rho_1$) is offered in the Bayesian nonparametric literature by the Enriched Dirichlet process \citep{wade2011enriched}, which in contrast induces \emph{conditional independence}, i.e., given $\rho_1$, observations at the second layer are assumed exchangeable if they belong to the same first-layer cluster and independent otherwise. 
However, enriched constructions do not allow to define a prior for $(\rho_1,\rho_2)$ with full support because they force second-layer clusters to be nested within first-layer clusters: if two items are assigned to distinct clusters at layer 1, they cannot be assigned to the same cluster at layer 2. 

To define a flexible and general learning mechanism for Bayesian clustering of multi-view or longitudinal data,  clusters defined by $\rho_1$  should be treated neither as almost irrelevant as under \emph{conditional exchangeability} nor as too informative as under the  \emph{conditional independence} of enriched constructions.  
Ideally, an appropriate learning mechanism would a priori favor at layer 2 a clustering configuration similar to layer 1, but not necessarily nested. A general framework to achieve this goal is currently missing and it is the main contribution of this work.
To this end, we introduce \emph{conditional partial exchangeability} (CPE), formalized by the following definition and discussed thereafter. Given a partition $\rho_1$ and the corresponding subject-
specific allocation variables $\bm{c}_1$, let us denote with $\mathcal{P}(n; \rho_{1})$ the set of permutations of $[n]$ that preserve $\rho_{1}$, i.e. $\sigma \in \mathcal{P}(n; \rho_{1n})$ if and only if $\sigma$ is a permutation of $n$ elements such that $c_{1\sigma(i)} = c_{1i}$, for any $i \in [n]$.
\begin{defi}[Conditional partial exchangeability]\label{def:CPE}
	Given a (marginally) exchangeable sequence $(X_{2i})_{i\geq1}$ and a collection of coherent\footnote{the collection of partitions $(\rho_{1n})_{n\geq1}$ is said to be coherent if for any $n$, $\rho_{1n}$ can be obtained from $\rho_{1n+1}$ by removing the object $n+1$.} random partitions $(\rho_{1n})_{n\geq1}$, where $\rho_{1n}$ is a partition of $[n]$,  $(X_{2i})_{i\geq1}$ is said to be conditionally partially exchangeable (CPE) with respect to $(\rho_{1n})_{n\geq1}$ if and only if, for any $n\geq1$,  for any realization of $\rho_{1n}$, and for any $(i_1,\ldots,i_\ell), (j_1,\ldots,j_\ell) \subset [n]$, the following two conditions are satisfied
    \begin{itemize}
	\item[c-i)]$p(X_{21},\ldots,X_{2n} \mid \rho_{1n}) = p(X_{2\sigma(1)},\ldots,X_{2\sigma(n)}\mid\rho_{1n}),\enskip \text{ for any } \sigma \in \mathcal{P}(n; \rho_{1n})$;
	\item[c-ii)] $p(X_{2i_1},\ldots,X_{2i_\ell} \mid c_{1i_1}=\ldots=c_{1i_\ell}) = p(X_{2j_1},\ldots,X_{2j_\ell} \mid c_{1j_1}=\ldots=c_{1j_\ell})$.
    \end{itemize}
\end{defi}
To fully understand  Definition~\ref{def:CPE}, we need first to consider the fundamental differences between CPE (introduced here) and partial exchangeability \citep[as presented by ][and quickly reviewed in Section S1 of the Supplement]{deFinetti1938}. Importantly, partial exchangeability is a condition on the marginal distributions of a sequence of observations, thereby defining a class of sequences. In contrast, CPE specifies a condition on the dependence between a sequence of exchangeable observations and a random partition, thereby defining a class of dependence relationships. Moreover, considering the conditions required for the conditional law of the observations, CPE differs from partial exchangeability not only in its conditional nature, imposed by condition (c-i), but also in its requirement for marginal invariance, imposed by condition (c-ii). Condition (c-ii) entirely lacks an analog in the definition of partial exchangeability. For a more intuitive understanding of CPE: assume the observations are partitioned according to a clustering configuration at layer 1. Conditioning on this layer-1 partition, the observations exhibit partial exchangeability at layer 2, i.e., in layer 2, exchangeability holds among observations that are co-clustered in layer 1, but not necessarily across different layer-1 clusters (c-i). Moreover, CPE requires that groups of observations corresponding to different layer-1 clusters preserve the same within-cluster marginal distribution at the subsequent layer (c-ii), e.g., $p(X_{21} \mid c_{11}) = p(X_{22}\mid c_{12}) $ and  $p(X_{21}, X_{22} \mid c_{11} =  c_{12} ) = p(X_{23}, X_{24} \mid c_{13} = c_{14})$.

For the sake of notation, in the following, we omit the subscript $n$ when denoting the partition.
 The next theorem shows that CPE allows the preservation of subjects' identities when moving from one layer to another.
\begin{thm}[Subjects' identity across layers]\label{thmfirst}
	If $(X_{2i})_{i\geq1}$ is conditionally partially exchangeable with respect to  $\rho_{1}$, then, for any measurable $A$ 
\begin{itemize}
    \item [s-i)]$\mathbb{P}( (X_{2i}, X_{2j})\in A^2 \mid c_{1i} = c_{1j} \neq c_{1k} ) \geq \mathbb{P}((X_{2i}, X_{2k})\in A^2 \mid c_{1i} = c_{1j}\neq c_{1k})$;
    \item[s-ii)] in general, $p( X_{2i}, X_{2j} \mid c_{1i} \neq c_{1j} ) \neq 	p( X_{2i} \mid c_{1i} \neq c_{1j} ) 	p(  X_{2j} \mid c_{1i} \neq c_{1j} ) $,  
\end{itemize}
where a strict inequality in s-i) is achievable as long as $(X_{2i})_{i\geq1}$ is not conditionally exchangeable with respect to $\rho_1$.
\end{thm}
All proofs can be found in the Appendix.
To preserve subjects' identity, we will always require that the inequality s-i) in Theorem~\ref{thmfirst} is strict for any measurable $A$: $0<\mathbb{P}(X_{2i}\in A)<1$. This ensures the non-degeneracy of  \emph{conditional exchangeability} in \eqref{eq:condexc} for which the left and right terms in s-i) would be equal. Moreover, we will also require the inequality in s-ii) to hold, and, thus, the non-degeneracy of \emph{conditional independence} of enriched constructions.   
Nonetheless, CPE recovers \emph{conditional exchangeability} \citep[as in the models in ][]{quintana2022dependent} and \emph{conditional independence} (as in enriched constructions) as limiting cases of minimal and maximal preservation of subjects' identities, corresponding to $\mathbb{P}( (X_{2i}, X_{2j})\in A^2 \mid c_{1i} = c_{1j} \neq c_{1k} ) = \mathbb{P}((X_{2i}, X_{2k})\in A^2 \mid c_{1i} = c_{1j}\neq c_{1k})$ and $p( X_{2i}, X_{2j} \mid c_{1i} \neq c_{1j} ) = 	p( X_{2i} \mid c_{1i} \neq c_{1j} ) 	p(  X_{2j} \mid c_{1i} \neq c_{1j} ) $, respectively.

The following propositions illustrate how existing Bayesian clustering approaches do satisfy or not CPE. 

\begin{prp}[Temporal random partition model]\label{prop:t-RPM}
	If $(X_{1i},\ldots, X_{Ti})_{i=1}^n$ follows the temporal random partition model (t-RPM) of \cite{page2022dependent} - Section 2, then $(X_{ti})_{i\geq1}$ is conditionally partially exchangeable with respect to $\rho_{t-1}$, but not conditionally exchangeable.  
\end{prp}

\begin{prp}[Separately exchangeable NDP-CAM]\label{prop:separate}
	If $(X_{1i},\ldots, X_{Ji})_{i\geq1}$ follows the separate exchangeable random partition model of \cite{lin2021separate} - Section 3, then, for any $j$ and $j'$, $(X_{j i})_{i\geq1}$ is conditionally partially exchangeable with respect to $\rho_j'$ but not conditionally exchangeable.  
\end{prp}

\begin{prp}[Dependent Dirichlet processes]\label{prop:DDP}
	If $(X_{1i}, X_{2i})_{i\geq1}$ follows a mixture model with mixing probabilities provided by dependent processes of the type described in \cite{maceachern2000dependent} and \cite{quintana2022dependent} = Section 2, then, conditionally on $\rho_1$, $(X_{2i})_{i\geq1}$ is conditionally exchangeable. 
\end{prp}

Finally, for the remainder of this work, we demonstrate that the strength of CPE extends beyond providing a condition to preserve subjects’ identities in multi-view probabilistic clustering. Instead, it can also serve as a \emph{constructive} definition that, due to its conditional formulation, facilitates the development and analysis of various clustering processes while ensuring analytical and posterior computational tractability.  

\section{The class of telescopic clustering models}\label{s:telmodel}
\subsection{A general telescopic clustering model}

First-layer observations $(X_{1i})_{i=1}^n$ are assumed to be distributed according to a mixture model \citep{ferguson1983bayesian,lo1984class}:
\begin{equation}
	\label{eq:firstlayergen}
	X_{1i}\mid \tilde p_{1} \overset{iid}{\sim} \int_{\Theta_1} k_1(X_{1i}, \theta) \, \tilde p_{1}(\dd\theta), \qquad  \text{for } i=1,\ldots, n,\\
\end{equation}
where $k_1(\cdot, \cdot)$ is a  kernel defined on $(\mathbb{X}_1,\Theta_1)$, $\tilde p_{1}$ is an almost-surely discrete random probability, i.e., $\tilde p_{1} \overset{a.s.}{=}  \sum_{m=1}^{M} w_m \delta_{\theta^{\star}_m}$, with $M\in\mathbb{N}\cup\{+\infty\}$ and $(w_m,\theta^{\star}_m)_{m=1}^M$ random variables such that $\sum_{m=1}^M w_m \overset{a.s.}{=} 1$. In the following, for notational convenience, the set $[M]:=\{1,\ldots,M\}$  denotes the set of the first $M$ natural numbers, when $M$ is finite, and the set of the natural numbers $\mathbb{N}$, when $M =\infty$.   Model \eqref{eq:firstlayergen} can be rewritten in   terms of the allocation 
vector  $\bm{c}_1 = (c_{11},\ldots,c_{1n})$, $i=1,\ldots,n$, defined in Section~\ref{s:mainCPE}:
$	X_{1i}\mid c_{1i}=m, \bm{\theta}^{\star}\overset{ind}{\sim} k_1(X_{1i}; {\theta^{\star}_{m}})
$.
%The ties in the vector of latent parameters $\bm{c}_1$  determine the partition $\rho_1$ of the observations into different clusters.
In the following, we assume that the subject-specific allocation variables $\bm{c}_1$ and the independent and identically distributed cluster-specific parameters $\bm{\theta}^{\star} = (\theta^{\star}_m)_{m=1}^M$ are apriori independent so that the corresponding mixing random probability $\tilde p_{1}$ belongs to the class of species sampling processes  \citep{pitman1996some}. 
To satisfy CPE, the second-layer conditional model is defined as
\begin{equation}
	\label{eq:secondlayergen}
	X_{2i} \mid c_{1i} = m, (\tilde p_{21},  \ldots, \tilde p_{2M}) \overset{ind}{\sim} \int_{\Theta_2} k_2(X_{2i}, \theta)\, \tilde p_{2m}(\dd \theta), \qquad \text{for } i=1,\ldots,n,
\end{equation}
where $k_2$ is a kernel defined on $(\mathbb{X}_2,\Theta_2)$, $M\in\mathbb{N} \,\cup \,\{+\infty\}$ is the number of mixture components at the first layer, $(\tilde p_{21},  \ldots, \tilde p_{2M})$ is a vector of (possibly dependent) almost-surely discrete and exchangeable random probability measures. Thus,  when $M=\infty$, $(\tilde p_{21},  \ldots, \tilde p_{2M})$ is  a countably infinite number of probability measures indexed by $\mathbb{N}$.
Bringing everything together, we arrive at the following definition for this class of models.
\begin{defi}[Telescopic clustering model]\label{def:modelclass}
	A random matrix $(X_{1i},X_{2i})_{i\geq 1}$ taking values in $(\mathbb{X}_1\times\mathbb{X}_2)^{\infty}$ is said to follow a telescopic clustering model (with two layers) if it admits the following representation:
	\[
	X_{1i}\mid \tilde p_{1} \overset{iid}{\sim} \int_{\Theta_1} k_1(X_{1i}, \theta) \, \tilde p_{1}(\dd\theta), \qquad  \text{for } i=1,2,\ldots
	\]
	\[
	X_{2i} \mid c_{1i} = m, (\tilde p_{21}, \ldots, \tilde p_{2M}) \overset{ind}{\sim} \int_{\Theta_2} k_2(X_{2i}, \theta)\, \tilde p_{2m}(\dd \theta), \qquad \text{for } i=1,2,\ldots
	\]
	with 
$	\tilde p_{1} \sim P_1$ and $(\tilde p_{21}, \ldots, \tilde p_{2M})  \sim P_2$, and
	where $k_1$ and $k_2$ are  kernels defined on $(\mathbb{X}_1,\Theta_1)$ and $(\mathbb{X}_2,\Theta_2)$, respectively;
	$\bm{c}_1 =( c_{11},\ldots, c_{1n})$ is a configuration of the allocation variables corresponding to the random partition $\rho_1$ induced by the marginal mixture model of $(X_{1i})_{i=1}^n$;
$M \in \mathbb{N}\cup\{+\infty\}$ is the number of mixture components in the marginal model of $(X_{1i})_{i=1}^n$;
 the prior $P_1$ is such that $\tilde p_{1}$ is an almost-surely discrete random probability measure;
the prior $P_2$ is such that $(\tilde p_{21}, \ldots, \tilde p_{2M})$ are almost-sure discrete (possibly dependent) exchangeable random probability measures. 
\end{defi}

A specific model is then obtained when the prior distributions $P_1$ and $P_2$ for $\tilde p_{1}$ and  $(\tilde p_{21}, \ldots, \tilde p_{2M})$, respectively, are chosen. 

The next theorem provides the hierarchical representation of the joint model for the data matrix. 

\begin{thm}[Telescopic clustering - joint representation] \label{thm:jointmix}
	If $(X_{1i},X_{2i})_{i\geq 1}$ follows a telescopic clustering model with two layers, as  in Definition~\ref{def:modelclass}, then, for $i=1,2,\ldots$, there exist $\theta_i$, $\xi_i$, and $\tilde p$, such that
	\begin{align*}
		(X_{1i}, X_{2i})\mid (\theta_i, \xi_i)\overset{ind}{\sim}k_1(X_{1i}, \theta_i)k_2(X_{2i}, \xi_i), \qquad	(\theta_i, \xi_i) \mid \tilde p \overset{iid}{\sim}\tilde p \overset{a.s.}{=} \sum_{m=1}^M \sum_{s=1}^S w_m q_{ms} \delta_{(\theta^{\star}_m,\,\xi^{\star}_s)}.
	\end{align*}
\end{thm}
 From the conditional construction of telescopic clustering models, it is trivial to prove that they satisfy CPE. In Theorem~\ref{thm:jointmix} we prove that the joint observations $(X_{1i}, X_{2i})_{i\geq1}$ are exchangeable across $i$ and that their law is \emph{Kolmogorov-consistent} in $n$. The latter condition is sometimes referred to as \emph{marginal invariance} \citep{dahl2017random} or \emph{projectivity} \citep{betancourt2022random} and, roughly speaking, implies that the model is suitable for drawing inferences on an infinite population since the distribution employed to model a finite sample $n$ is extendable.  
Here, \emph{Kolmogorov consistency} follows directly from the conditional-i.i.d. sampling of $(\theta_i, \xi_i)$ in Theorem~\ref{thm:jointmix} and de Finetti's theorem \citep{de1937prevision} for infinite exchangeable sequences of random variables.
 
Finally, if a global clustering structure (i.e., based on all layers) is of interest, the telescopic model still provides appropriate inference. Indeed, in telescopic clustering, global clusters are defined as the common refinement of the partitions at different layers, i.e., two subjects belong to the same global cluster if they belong to the same cluster at all layers. 
Still, the main goals of telescopic clustering models is different: (1) provide also, possibly different, clustering configuration at each layer, (2) allow global clusters to share all or a subset of latent parameters at any layer \citep[cfr.,][]{petrone2009hybrid}, (3) allow
more flexible transfer of information across features, which translates into better inferential performance (see Section~\ref{s:sim}),  (4) allow investigating dependence between features in terms of dissimilarities between clustering configurations at different layers. The latter point is more extensively described in the next section. 

\subsection{Measures of telescopic dependence}\label{ss:measuredep}
The class of models described above allows for a bi-variate clustering configuration of the same observational units taking into account within-subject dependence. In this section, four dependence measures between clustering configurations (at the different layers) are presented.  The measure of telescopic dependence and the telescopic adjusted Rand index are novel measures of dependence that capture specific properties of telescopic clustering models, while the remaining two are widely used measures: the expected Rand index and the expected Binder loss.

In telescopic clustering models, the probability of any two subjects being clustered together at layer 2 depends on whether they were clustered together at layer 1, with a higher probability if they were already clustered together than if they were not. This result follows directly from the fact that, in partially exchangeable mixture models with equal marginal distributions, the probability of ties within a group is always at least as high as the probability of ties across groups \citep[for details, see,][]{ascolani2024nonparametric,franzolini2022ondependent,franzolini2023multivariate}.
In light of this, we define a conditional measure of similarity between $\rho_1$ and $\rho_2$ as a normalized difference between conditional probabilities of ties.  
\begin{defi}[Measure of telescopic dependence]
	Given two random partitions $\rho_1$ and $\rho_2$ of the same subjects,
	\[
	\text{$\tau$} = \frac{\mathbb{P}[c_{2i} = c_{2j}  \mid c_{1i} = c_{1j}] - \mathbb{P}[c_{2i} = c_{2j} \mid c_{1i} \neq c_{1j}]}{\mathbb{P}[c_{2i} = c_{2j}  \mid c_{1i} = c_{1j}]}
	\]
	is called measure of telescopic dependence between $\rho_1$ and $\rho_2$.
\end{defi}
By definition, $\text{$\tau$} \in [0,1]$ and $\text{$\tau$} = 1$ if and only if$\mathbb{P}[c_{2i} = c_{2j} \mid c_{1i} \neq c_{1j}] = 0$, while $\text{$\tau$} = 0$ if and only if$\mathbb{P}[c_{2i} = c_{2j}  \mid c_{1i} = c_{1j}] = \mathbb{P}[c_{2i} = c_{2j} \mid c_{1i} \neq c_{1j}]$.
It is immediate to show that when $\rho_1$ and $\rho_2$ are independent, then $\text{$\tau$} = 0$.  On the other hand, under the enriched Dirichlet process $\text{$\tau$} = 1 $, indicating maximum telescopic dependence, while in our framework  $\text{$\tau$} \in [0,1] $. This is due to the fact that in telescopic clustering $\mathbb{P}[c_{2i} = c_{2j} \mid c_{1i} \neq c_{1j}]$ can be positive, while in the enriched Dirichlet process the same probability is equal to zero for any value of the hyperparameters, resulting in a smaller support for the joint prior of the partitions. 
The measure $\text{$\tau$}$ of telescopic dependence is an asymmetric measure, which is computed conditionally on the allocation at layer 1.

 Denote with $\Pi(n)$ the space of partitions of $n$ elements and with $p(\rho_1, \rho_2)$ the joint probability law of the two clustering configurations induced by a telescopic clustering model, which we name \emph{telescopic exchangeable partition probability function} (t-EPPF). In general, the t-EPPF has full support on the space of bi-variate clustering configurations $\Pi(n)^2$, while still encoding dependence between clustering configurations. 
In the following, we consider the expected Rand index (ER) and the expected Binder loss (EB) between $\rho_1$ and $\rho_2$, defined respectively as
$$ER = {n\choose 2} ^{-1}\int\limits_{\Pi(n)^2} [a(\rho_1, \rho_2) +  b(\rho_1, \rho_2)] \, \dd \, p(\rho_1, \rho_2)$$ 
$$EB = \int\limits_{\Pi(n)^2} [c(\rho_1, \rho_2) +  d(\rho_1, \rho_2)] \, \dd \, p(\rho_1, \rho_2) $$
where $a$, $b$, $c$, and $d$ are functions of the partitions: $a$ returns the number of pairs of observations clustered together both at layer 1 and 2, $b$ the number of pairs clustered together neither at layer 1 nor 2, $c$ the number of pairs clustered together at layer 1, but not at layer 2, and $d$ the number of pairs clustered together at layer 2 but not at layer 1. 

\begin{prp}[Dependence measures as functions of the number of clusters]\label{prop:exp_rand}
	In a telescopic clustering model, a priori  
	\begin{align*}
		\text{$\tau$} = \frac{\mathbb{P}(K_{22} = 1 \mid K_{12} = 1) - \mathbb{P}(K_{22} = 1 \mid K_{12} = 2) }{\mathbb{P}(K_{22} = 1 \mid K_{12} = 1)}
	\end{align*}
    \begin{align*}
		ER =\mathbb{P}(K_{12} = K_{22} ), \quad&\quad  EB = {n\choose 2}\mathbb{P}(K_{12} \neq K_{22} ) 
	\end{align*}
	where $K_{\ell n}$ denotes the number of clusters at layer $\ell$ in a sample of $n$ subjects. 
\end{prp}

As noted by \cite{hubert1985comparing}, when the Rand index is used to compare random partitions, its expected value is not 0 in case of independence of the partitions. In a telescopic clustering, when $\rho_1$ and $\rho_2$ are independent, the expected value of the rand index is
$ER^{\perp} =\, \sum_{\kappa =1}^2\mathbb{P}(K_{12} = \kappa)\mathbb{P}(K_{22} = \kappa) $
where $\perp$  denotes independence (see Proposition~\ref{prop:exp_rand} above). Thus, $ER^{\perp}$ is typically positive. In the same spirit as that of the adjusted Rand index \citep{hubert1985comparing}, we define a \emph{telescopic adjusted rand-index} that allows us to correct for the randomness of the partitions. 
\begin{defi}[Telescopic adjusted Rand index]
	The telescopic adjusted Rand index between $\rho_1$ and $\rho_2$ is defined as 
	\[
	TARI = \frac{[a(\rho_1, \rho_2) +  b(\rho_1, \rho_2)] - ER^{\perp}}{1 - ER^{\perp}}
	\]
\end{defi}
It is trivial to prove that, in the case of independence, the a priori expected value of the $TARI$ equals 0.

\subsection{Extension to $L$ layers using polytrees}
\label{ss:Llayers}
\begin{figure}[H]
	\begin{center}
	\resizebox{0.85\textwidth}{!}{
		\begin{tikzpicture}
			\tikzstyle{main}=[circle, minimum size = 13mm, thick, draw =black!80, node distance = 10mm]
			\tikzstyle{connect}=[-latex, thick]
			\tikzstyle{box}=[rectangle, draw=black!100]
			\node[main] (O1) {$\rho_1$};
			\node[main] (O3) [ right=of O1] {$\rho_2$};
			\node[main] (O5) [ right=of O3] {$\rho_3$};
			\node[main] (O6) [ right=of O5, draw=white!100] {$\ldots$};
			\node[main] (O7) [ right=of O6] {$\rho_t$};
			\node[main] (O8) [ right=of O7, draw=white!100] {$\ldots$};
			\node[main] (O9) [ right=of O8] {$\rho_T$};
			\node[main] (O10) [ right=of O9, draw=white!100] {$\ldots$};
			%%%%%%%%%%%%%%%%%%%%%%%%%%%%%%%%%%%%%
			\path (O1) edge [connect] (O3);
			\path (O3) edge [connect] (O5);
			\path (O5) edge [connect] (O6);
			\path (O6) edge [connect] (O7);
			\path (O7) edge [connect] (O8);
			\path (O8) edge [connect] (O9);
			\path (O9) edge [connect] (O10);
	\end{tikzpicture}}
	\caption{\label{fig:longitudinalgraph} Layer dependence for longitudinal data.}
	\end{center}
\vspace{-0.5cm}
\end{figure}
The class of telescopic models as presented in the previous sections defines a prior distribution for the joint law of two partitions, $\rho_1$ and $\rho_2$, through the product
$p(\rho_1)\,p(\rho_2\mid \rho_1)$
where $p(\rho_1)$ and $p(\rho_2\mid \rho_1)$ are used to denote the marginal law of the partition $\rho_1$ and the conditional law of the partition $\rho_2$, respectively. 

The main advantage and novelty of this class of models lie in how the dependence between the two partitions is defined through the CPE, which ultimately specifies a one-way relationship from $\rho_1$ to $\rho_2$, denoted in the following as $\rho_1\rightarrow\rho_2$.

\begin{wrapfigure}{l}{0.23\textwidth}
	\resizebox{0.23\textwidth}{!}{
		\begin{tikzpicture}
			\tikzstyle{main}=[circle, minimum size = 13mm, thick, draw =black!80, node distance = 10mm]
			\tikzstyle{connect}=[-latex, thick]
			\tikzstyle{box}=[rectangle, draw=black!100]
			\node[main] (O1) {$\rho_X$};
			\node[main] (O3) [below right=of O1] {$\rho_Z$};
			\node[main] (O5) [below left=of O1] {$\rho_Y$};
			%%%%%%%%%%%%%%%%%%%%%%%%%%%%%%%%%%%%%
			\path (O1) edge [connect] (O3);
			\path (O1) edge [connect] (O5);
	\end{tikzpicture}}
	\caption{\label{fig:multi-view}  Triangular dependence for three layers.}
\end{wrapfigure}
A straightforward way to extend the modeling strategy to any number of layers is by combining multiple pairwise relationships in a polytree.
For instance, in the context of longitudinal data, where different measurements are collected at different time points a Markovian structure across different layers can be imposed. The resulting telescopic clustering model is then obtained assuming CPE between $\Pi_t$ and $\Pi_{t+1}$ for any $t \in \mathbb{N}$, i.e., 
$p(\rho_t,\,t=1,2,\ldots) = p(\rho_1)\prod_{t=2}^{\infty}p(\rho_t\mid \rho_{t-1})$.
See Figure~\ref{fig:longitudinalgraph}.

A second extension that we consider in this work involves combining the dependence across three sets of features through the triangular graph represented in Figure~\ref{fig:multi-view}. In this setting, given the clustering configuration of $X$, which is the response variable of main interest, the goal is to also infer additional clustering configurations for two other sets of variables:  $Y$ and $Z$. Then, the t-EPPF of the model is given by  
$p(\rho_X,\rho_Y,\rho_Z) = p(\rho_X)p(\rho_Y\mid \rho_X)p(\rho_Z\mid \rho_X)$. The polytrees strategy is based on a partial ordering of the different layers, due to the fact that each node in the graph can have at most one parent node and the multivariate dependence across layers is obtained by combining pairwise dependence only. Nonetheless, in the structure in Figure~\ref{fig:multi-view}, the CPE induces a mutual (undirected) dependence between $\rho_Y$ and $\rho_Z$ in the sense that, when $p(\rho_Y\mid \rho_X) = p(\rho_Z\mid \rho_X)$, the conditional law of $\rho_Z$ given $\rho_Y$ is the same as the conditional law of $\rho_Y$ given $\rho_Z$.

\FloatBarrier 
\section{A telescopic model with infinite number of labels}\label{ss:tHDP}
Hierarchical constructions for dependent processes, initially introduced in \cite{teh2006teh}, offer a powerful framework for modeling dependence across random distributions. In \cite{teh2006teh}, the construction is based on the Dirichlet process and it was further extended to encompass more general processes in \cite{camerlenghi2019distribution}, \cite{argiento2020hierarchical}, and \cite{bassetti2020hierarchical}. 
We employ this construction to build up the telescopic mixtures with hierarchical Dirichlet processes (t-HDP), where we set as prior for the first-layer random probability $\tilde p_{1}$ an HDP, which defines  the law of a single process \citep[for details and generalization of this prior, see][]{camerlenghi2018bayesian} such that
\begin{equation}
	\label{eq:HDPmarg}
	\begin{aligned}
		\tilde p_{1} \mid \gamma, \tilde p_0 \sim DP(\gamma, \tilde p_0),\qquad	\tilde p_0 \mid \gamma_0 \sim DP(\gamma_0, P_{\theta}),
	\end{aligned}
\end{equation} 
while the second-layer conditional law is 
\begin{equation}
	\begin{aligned}
		X_{2i} \mid \bm{c}_1, (\tilde p_{21}, \tilde p_{22} \ldots, ) &\overset{ind}{\sim} \int f(X_{2i},\theta) \tilde p_{2c_i}(\dd \theta)\\
		\tilde p_{2m} \mid  \alpha,  \tilde q_0 \overset{iid}{\sim} DP(\alpha, \tilde q_0), \quad&\quad 	\tilde q_0  \mid \alpha_0 \sim DP(\alpha_0, P_{\xi}),\\
	\end{aligned}
\end{equation}
where $DP(\alpha, P)$ denotes a Dirichlet process with concentration parameter $\alpha$ and base distribution $P$.
Consider a specific partition $\rho_1$ into $K_{1n}$ sets of numerosities $n_1,\ldots,n_{K_{1n}}$ for the first-layer partition. Then,  we have \citep[see,][]{camerlenghi2018bayesian}
\begin{equation}
	\label{eq:EPPFHDP}
	\mathbb{P}[\rho_1 = \rho] = \frac{\gamma_0 ^ {K_{1n}}}{ (\gamma)^{(n)}} \sum_{\bm{\ell}}\frac{\gamma^{|\bm{\ell}|}}{(\gamma_0)^{(|\bm{\ell}|)}} \prod_{m=1}^{K_{1n}} (\ell_m - 1)! |s(n_m, \ell_m)|
\end{equation}
where $|s(n, k)|$ denotes the signless Stirling number of the first kind and the sum in \eqref{eq:EPPFHDP} runs over all vectors $(l_1,\ldots, l_{K_{1n}})$ such that $l_m \in [n_m]$ and $(\gamma)^{(n)} = \Gamma(\gamma + n) / \Gamma(\gamma)$, where $\Gamma(x)$ denote the Gamma function in $x$.
The conditional law of the partition at layer 2, given $\rho_1$, is 
\begin{equation}
	\label{eq:pEPPFHDP}
	\mathbb{P}[\rho_2 = \rho\mid \rho_1] = \frac{\alpha_0^{K_{2n}}}{\prod_{m=1}^{K_{1n}}(\alpha)^{(n_m)}}
	\sum_{\bm{t}}\frac{\alpha^{|\bm{t}|}}{(\alpha_0)^{(|\bm{t}|)}} \prod_{s=1}^{K_{2n}} ( t_{\cdot s} - 1)! 
	\prod_{m=1 }^{K_{1n}}  |s(n_{ms}, t_{ms})|
\end{equation}
where the sum runs over all matrices $K_{1n} \times K_{2n}$, whose generic element $t_{ms}$ belong to $[n_{ms}]$ provided that $n_{ms} \geq 1$, and is equal to 1 when $n_{ms} = 0$.  Moreover, $t_{\cdot s} =\sum_{m}^{K_{1n}} t_{ms}$. See  \cite{camerlenghi2019distribution}.

\begin{thm}[t-EPPF in the t-HDP model]\label{thm:tEPPFinfinite}
	Given a telescopic mixture model with hierarchical Dirichlet processes and two layers, $p(\rho_1, \rho_2)$ is given by 
	\footnotesize
	\begin{align*}
		\frac{\gamma_0 ^ {K_{1n}} \alpha_0^{K_{2n}}}{ (\gamma)^{(n)}\prod\limits_{m=1}^{K_{1n}}(\alpha)^{(n_m)}} 
		\sum_{\bm{\ell},\bm{t}}\frac{\gamma^{|\bm{\ell}|}\alpha^{|\bm{t}|}}{(\gamma_0)^{(|\bm{\ell}|)} (\alpha_0)^{(|\bm{t}|)}} \left(\prod_{m=1}^{K_{1n}} (\ell_m - 1)! |s(n_m, \ell_j)|\right) \, \prod_{s=1}^{K_{2n}}  ( t_{\cdot s} - 1)! \prod_{m=1}^{K_{1n}}  |s(n_{ms}, t_{ms})| 
	\end{align*}
	\normalsize
\end{thm}

Starting from the expression of the t-EPPF, it is straightforward to compute the indexes of dependence introduced in Section~\ref{ss:measuredep}. 

\begin{cor}[Measures of dependence in the t-HDP model]
	\label{cor:D12_hdp}
	In a t-HDP, the measure $\text{$\tau$}$ of telescopic dependence is
	$\text{$\tau$} = \frac{\alpha_0}{\alpha_0 + \alpha + 1}$ and the expected Rand index is 
$ER = \frac{(1+\gamma_0+\gamma)(1 +\alpha_0 +\alpha) +  \gamma_0\, \alpha_0\,\gamma\, \alpha }{(\gamma_0+1)(\gamma+1)(\alpha_0 + 1)(\alpha+1)}.
$
\end{cor}
Thus, $\tau$ tends to $0$ as $\alpha$ tends to $\infty$ and to $1$ as $\alpha_0$ tends to $\infty$.
Finally, at layer 1, alternative priors might be more suitable depending on the application, such a non-hierarchical prior such as the classical Dirichlet Process \citep{ferguson1973bayesian}, or the Pitman-Yor process \citep[e.g.,][]{pitman1997two}. However, incorporating a hierarchical structure in the subsequent layers is essential for achieving conditional partial exchangeability.

\section{A telescopic model with random number of labels}\label{s:MFMmodel}
The t-HDP model introduced in the previous section assumes that the number of sub-populations (or components) in the mixtures equals infinity, which is a classical modelling assumption in Bayesian nonparametric mixtures models. Nonetheless, an alternative successful strategy consists in assuming that the number $M$ of sub-populations is almost-surely finite and placing a prior over $M$. The second telescopic model introduced here  lies within this framework.    
The prior for the first-layer random probability $\tilde p_{1}$ is defined by
\label{sec:sec}
\begin{equation}
	\label{eq:firstmix}
	\begin{aligned}
		&\tilde p_{1} = \sum\limits_{m=1}^{M} w_m \delta_{\theta^{\star}_m}\\ 
		\bm{w}=(w_1,\ldots,w_M) \mid M \sim P_w, \qquad &\theta^{\star}_m \mid M \overset{iid}{\sim} P_{\theta}, \quad \text{for } m=1,\ldots, M, \qquad &	M \sim P_M,
	\end{aligned}
\end{equation}
where $\bm{w}$ and $\bm{\theta^{\star}}=(\theta^{\star}_1,\ldots,\theta^{\star}_M)$ are independent and $P_M$ has support on the set of the natural numbers $\mathbb{N}$. The resulting marginal model for the first layer is a finite mixture with a random number of components \citep{nobile1994bayesian,miller2018mixture,argiento2019infinity}.
Depending  on the choice of $P_w$ different finite-dimensional prior processes can be employed as priors for the finite mixture construction. In the following, we focus on the  Dirichlet distribution as prior for the weights, as it is the most popular in applications, i.e.,
$	\bm{w} = \left(w_1,\ldots,w_M\right)\mid M, \gamma \sim \text{Dirichlet}_{M}\left(\gamma, \ldots, \gamma\right)$.
Then, the conditional law of the second layer is defined employing a novel construction for the mixing random probability measures, whose formal construction is detailed in the following definition.
\begin{defi}[Unique-atom process]\label{def:UA-process}
	A vector of random probabilities $(\tilde p_1, \ldots, \tilde p_K)$ is a unique-atom process if they admit the following almost-sure discrete representation:
$\tilde p_m \overset{a.s.}{=} (1-Z) \, \delta_{\xi^{\star}_m}+ Z\, \tilde p_0$ for $m=1,\ldots,K$, with	$Z\sim\text{Bernoulli}(\omega)$,
	where 
 $\tilde p_0$ is an almost-surely discrete random probability,
 $\xi^{\star}_m \overset{iid}{\sim} P_{\xi}$, for $m = 1,\ldots,K$, and
 $\tilde p_0$, $(\xi^{\star}_m)_{m=1}^K$, and $Z$ are pairwise independent.
\end{defi}
In the following, we make use of unique-atom processes where the common $\tilde p_0$ in the previous definition is a random probability with a random (almost-surely finite) number of support points and Dirichlet weights, i.e.,  
$\tilde p_0 \overset{a.s.}{=} \sum_{s=1}^S q_{s} \delta_{\xi^{\star}_{0s}}$
with $S\sim P_S$, weights $q_{s}$ distributed accordingly to a symmetric Dirichlet distribution and $\xi^{\star}_{0s}\overset{iid}{\sim}P_{\xi}$. 
The rationale behind the construction in Definition~\ref{def:UA-process} is the following: when the random variable $Z = 0$, the clustering structure is kept constant from one layer to the next, while when $Z=1$, the clustering structure is estimated independently from the clustering arrangement at the previous layer. 
Employing unique-atom processes to build up CPE needed for telescopic clustering, we get the following second-layer specification 
\begin{equation}
	\label{eq:secondmix}
	\begin{aligned}
		X_{2i} \mid \bm{c}_1, \bm{q}, \bm{\xi}, S, Z &\overset{ind}{\sim} (1-Z) k_2(X_{2i} ;\xi^{\star}_{c_{1i}})+Z\,\sum\limits_{s=1}^{S} q_s k_2(X_{2i} ;{\xi^{\star}_{0s}}) \\
		\bm{q} = (q_{1},\ldots,q_{S}) \mid S , \alpha & \sim\text{Dirichlet}_S(\alpha, \ldots, \alpha)\\ 
		{\xi^{\star}_{0s}} \mid S  \overset{iid}{\sim} P_{\xi}, \qquad \xi^{\star}_{m}\mid K_{1n}&\overset{iid}{\sim} P_{\xi}, \qquad S \sim P_M,\qquad Z \sim \text{Bernoulli}(\omega).
	\end{aligned}
\end{equation}

The joint law of the two partitions is provided by the next theorem.

\begin{thm}[t-EPPF in the telescopic unique atom process]\label{thm:mEPPFfinite}
	Given a telescopic mixture with unique atom processes, the t-EPPF is
	\begin{equation*}
		\begin{aligned}
			p(\rho_1, \rho_2) =   &(1-\omega) V(n,K_{1n}) \, \prod_{m = 1}^{K_{1n}} \frac{\Gamma(\gamma + n_{m})}{\Gamma(\gamma)} \mathbbm{1}(\rho_1=\rho_2)  \\
			+ &\omega \,V(n,K_{2n}) \, \prod_{s = 1}^{K_{2n}} \frac{\Gamma(\alpha + \sum_{m=1}^{K_{1n}}n_{ms})}{\Gamma(\alpha)} V(n,K_{1n}) \, \prod_{m = 1}^{K_{1n}} \frac{\Gamma(\gamma + n_{m})}{\Gamma(\gamma)}
		\end{aligned}
	\end{equation*}
    where 
    $V(n,K) = \sum\limits_{M=1}^{+\infty} \frac{M_{(K)}}{(\gamma K)^{(n)}} p_M(M)$.
\end{thm}

\begin{cor}[Measures of dependence in the telescopic unique atom process]
	\label{cor:D12_CAWI}
	In a telescopic mixture with unique atom processes,
$\text{$\tau$} = \frac{1-\omega}{1 + \omega (\mathbb{E}[S]/\alpha - 1)}$ and 
$$ER = \frac{\mathbb{E}[M]}{\gamma}\left(1 - \omega + \omega \frac{\mathbb{E}[S]}{\alpha}\right) + 
	\frac{\mathbb{E}[M(M-1)]\gamma^2}{4\gamma^2+2\gamma}\left(1 - \omega + \omega \frac{\mathbb{E}[S(S-1)]}{4\alpha^2+2\alpha}\right). $$
\end{cor}
Thus, $\tau$ tends to $1$ as $\omega$ tends to 0 and to $0$ as $\omega$ tends to 1.

%%%%%%%%%%%%%%%%%%%%%%%%%%%%%%%%%%%%%%%%%%%%%%%%%%%%%%%%%%%%%%%%%%%%%%

\section{Algorithms for posterior inference}\label{s:algorithms}
Similarly to existing Bayesian mixture models, also in telescopic clustering models, posterior inference can be performed through either conditional or marginal Markov chain Monte Carlo (MCMC) algorithms. The conditional algorithms make use of representation theorems and also provide posterior samples of the underlying random probability measures \citep[see, for instance,][]{ishwaran2001gibbs, walker2007sampling}. Nonetheless, when the number of components is infinite, conditional algorithms typically require to rely on a truncated approximation of the underlying random probability measure. In contrast, the marginal algorithms are derived through marginalization of the random probability \citep[see, for instance,][]{neal2000markov}.

In the case of telescopic clustering models, marginal algorithms require evaluating the conditional law of the partition at the child nodes when sampling the cluster allocation at any given parent layer. However, evaluating this conditional law is typically computationally intensive, and introducing latent random variables to reduce the cost is not always straightforward. For example, in t-HDP models, the standard data augmentation provided by the Chinese restaurant franchise process \citep{teh2006teh} simplifies the conditional law of the partition to be evaluated but significantly slows down the mixing to unfeasible levels (for details, see Sections S3 and S4 of the Supplement). Therefore, enabling inference using marginal algorithms requires a tailored variable augmentation scheme for each telescopic clustering model. 

On the other hand, the conditional sampling scheme for the t-HDP model exhibits good mixing and significantly lower computational time per iteration, making posterior inference feasible and, importantly, easily adaptable to different poly-tree structures and prior choices (for details, see Sections S2, S3, and S6 of the Supplement).  

For these reasons, the results presented in the following sections are obtained via a truncated blocked Gibbs sampler. This algorithm is a conditional one, easier to generalize within the class of telescopic clustering models, provided that the full conditionals of the weights of the random probability measures are available. Unlike the marginal sampling scheme, it does not require model-specific data augmentation techniques, making it the preferred choice for this work. However, it relies on a truncated version of the random probability measures when the number of components is infinite, as in the t-HDP, and thus incurs a truncation error cost.
Thus, this approach incurs a truncation error cost. A promising direction for future research is the adaptation of such schemes using slice sampling techniques \citep{walker2007sampling,kalli2011slice}, as has been recently applied to the classical HDP in \citep{amini2019exact} and \cite{das2024blocked}. It is important to notice that the availability of conditional sampling schemes depends on the existence of (conditional) representation theorems and underlying random probabilities, which, thus, for telescopic clustering, are not only an analytical and probabilistic result but a fundamental computational tool. 
A detailed derivation of the sampling schemes, computational cost, and mixing performance for telescopic models are in Sections~S3, S4, and S7 of the Supplement.

%%%%%%%%%%%%%%%%%%%%%%%%%%%%%%%%%%%%%%%%%%%%%%%%%%%%%%%%%%%%%%%%%%%%%%
\section{Numerical studies}
\subsection{Simulation study}\label{s:sim}

\begin{table}[tb]
	\centering
	\resizebox{0.65\textwidth}{!}{
		\begin{tabular}{c||c|c|c|c||c|c|c|c}
			& \multicolumn{4}{c||}{Rand Index} & \multicolumn{4}{c}{\# Mistakes}\\
			Layer& k-means & t-HDP & LSBP & E-DP  & k-means & t-HDP & LSBP & E-DP \\
			\hline
			n.1&0.98&0.98&0.98&0.50&2&2&2&100\\
			n.2&0.98&1.00&0.98&0.90&2&0&2&10\\
			n.3&0.92&0.98&0.92&1.00&8&2&8&0\\
			n.4&0.98&1.00&0.98&0.92&2&0&2&17\\
			n.5&0.92&0.97&0.91&0.89&8&3&9&21\\
			n.6&0.97&0.98&0.97&0.86&3&2&3&31\\
			n.7&0.94&0.99&0.92&0.83&6&1&8&40\\
			n.8&0.95&1.00&0.95&0.79&5&0&5&44\\
			n.9&0.93&1.00&0.93&0.79&7&0&7&47\\
			n.10&0.91&0.99&0.89&0.75&9&1&11&54\\
			\hline
			average& 0.95& \textbf{0.99} & 0.83 & 0.82&5.2&\textbf{1.1}&5.7&36.4\\
	\end{tabular} }
	\caption{\label{tab:scenario3}Scenario 1, Rand indexes between the estimated and true clustering configurations and numbers of items allocated to the wrong cluster. } 
\end{table} 

Here we report results for a few simulations, additional simulation studies with different numbers of layers and misspecification can be found in Section S5 of the Supplement, together with additional results regarding the simulations described here below.

In the first simulation scenario (Scenario 1), we generate data on $n = 200$ items and $T = 10$ layers. At each layer, marginally we assume two clusters simulated from two univariate Normal distributions with unitary variance and centered in $0$ and $4$ respectively. From one layer to the next, 10 items (5\% of the total) are selected at random and moved to the other cluster.

\begin{wrapfigure}{l}{0.5\textwidth}
	\centering
	\includegraphics[width=\linewidth]{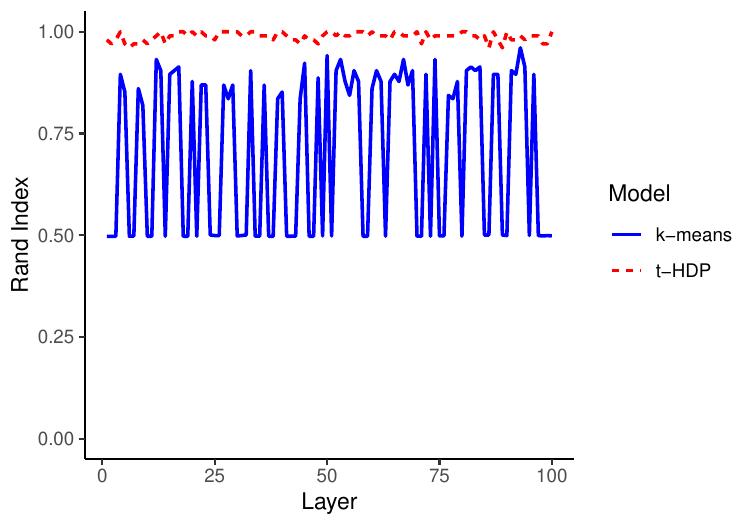}
	\caption{\label{sim4xx}Scenario 2. Rand indexes between the truth and the estimated configuration.}
	\label{fig:randsim5}
	\vspace{-\baselineskip}
\end{wrapfigure} We compare four methods: (i) k-means fitted independently at each layer, where the number of clusters is determined by the gap statistics \citep{tibshirani2001estimating}; (ii) the t-HDP's estimate; (iii) the estimate obtained with a logit stick-breaking process (LSBP) \citep{ren2011logistic}; and (iv) the estimate from an Enriched Dirichlet process (E-DP) \citep{wade2011enriched}.
For the LSBP, the layer's number is used as a covariate for both the weights and the atoms \citep[for more details and algorithms, see,][]{rigon2021tractable}. For models (ii)-(iv), we use a Gaussian kernel for the nonparametric mixture with a  Normal-InverseGamma for the mean and the variance as base measure. We report as a point estimate for the clustering configuration the one that minimizes the variation of information loss \citep{meila2007comparing}.
Table~\ref{tab:scenario3} summarises the results. 
The t-HDP model outperforms the competitors both consistently at each layer and overall. 
In Scenario 2, data for $T=100$ layers are simulated. At each layer, there are two clusters with 100 observations each. At layer 1,  data are sampled from	$X_{1i}\mid c_{1i}\overset{ind}{\sim}\mathcal{N}(0,1) \mathbbm{1}( c_{1i} = 1) + 
\mathcal{N}(3,1) \mathbbm{1}( c_{1i} = 2) $
Then, from layer $\ell$ to layer $\ell+1$, $2\%$ of the observations are selected at random and moved to the other cluster.
Figures~\ref{sim4xx} and \ref{fig4yy} summarize the results of the t-HDP model and independent k-means clustering, where again the t-HDP outperforms k-means. Posterior estimates are obtained by minimizing the variation of information loss \citep{meila2007comparing} and by employing the gap statistics \citep{tibshirani2001estimating}. 

\begin{figure}[H]
	\begin{subfigure}{.28\textwidth}
		\includegraphics[width=\linewidth]{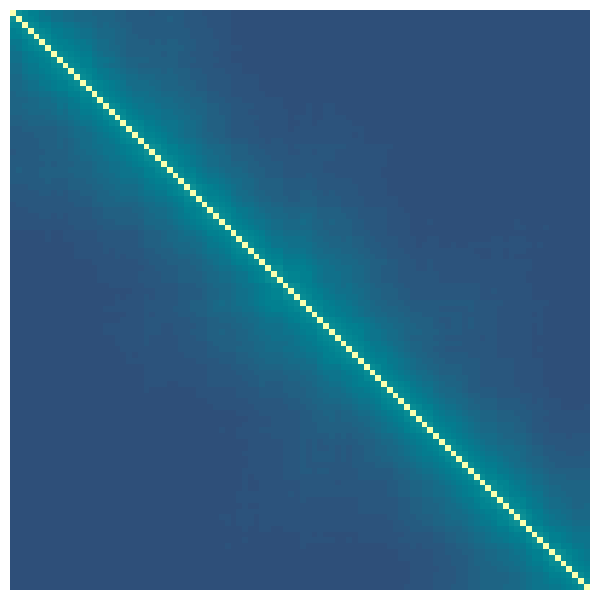}
		\caption{True dependence}
	\end{subfigure}
	\hfill
	\begin{subfigure}{.28\textwidth}
		\includegraphics[width=\linewidth]{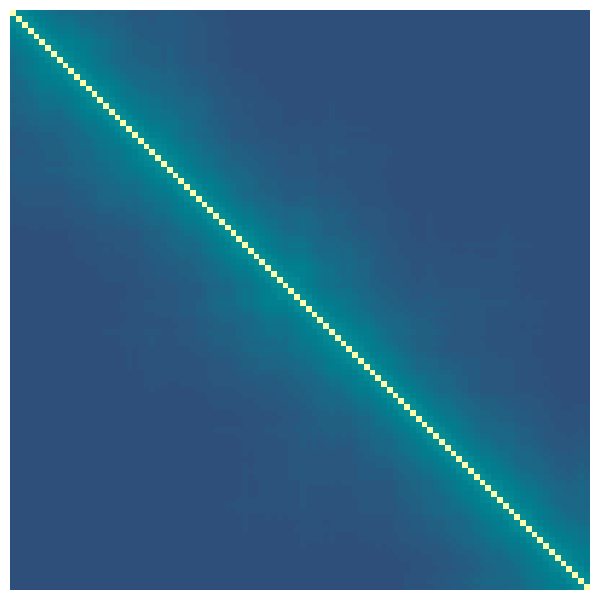}
		\caption{t-HDP estimate}
	\end{subfigure}
	\hfill
	\begin{subfigure}{.28\textwidth}
		\includegraphics[width=\linewidth]{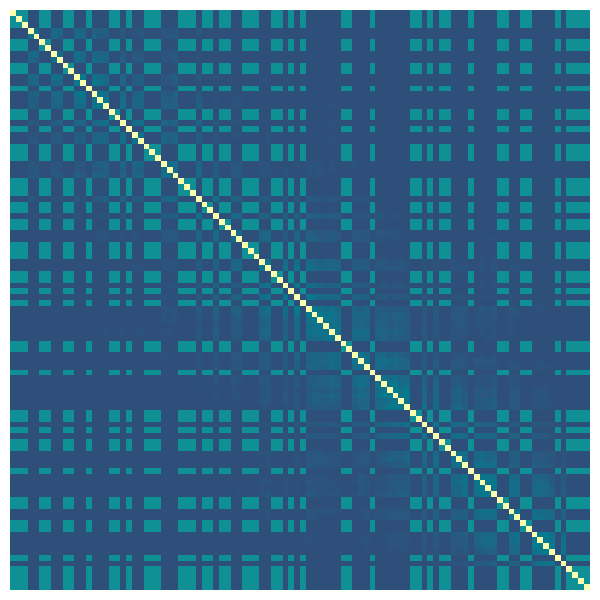}
		\caption{k-means estimate}
	\end{subfigure}
	\caption{\label{fig4yy} Simulation study: results for Scenario 2. Pairwise Rand indexes between any couple of layers for (a) the true clustering configurations; (b) the t-HDP model; (c) k-means. }
\end{figure}
\subsection{An application to childhood obesity}\label{s:metabolites}

\begin{figure}[htb!]
	\resizebox{\textwidth}{!}{
		\begin{tikzpicture}
			\tikzstyle{main}=[circle, minimum size = 13mm, thick, draw =black!80, node distance = 10mm, line width = 1 mm]
			\tikzstyle{connect}=[-latex, thick]
			\tikzstyle{box}=[rectangle, draw=black!100]
			\node[main] (1B) [minimum size = 10cm, fill=blue!20, align=center] {\HUGE{{28\%}} \\ \\ \includegraphics[width=\textwidth]{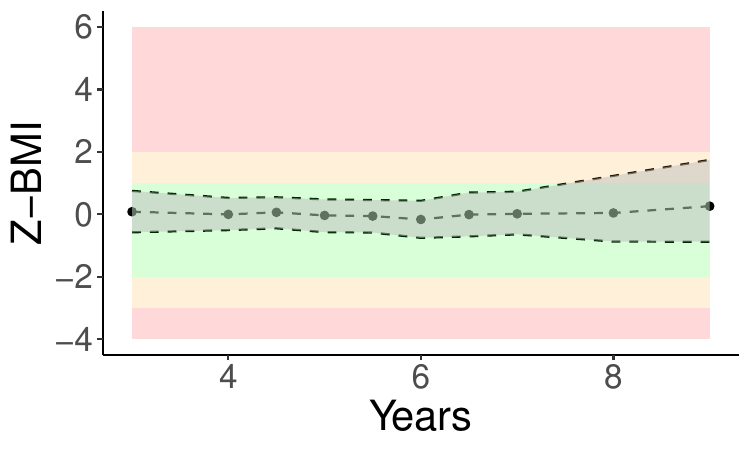}};
			\node[main] (1A) [left= 3 cm of 1B, minimum size = 10cm, fill=blue!20, align=center] {\HUGE{{26\%}} \\ \\ \includegraphics[width=\textwidth]{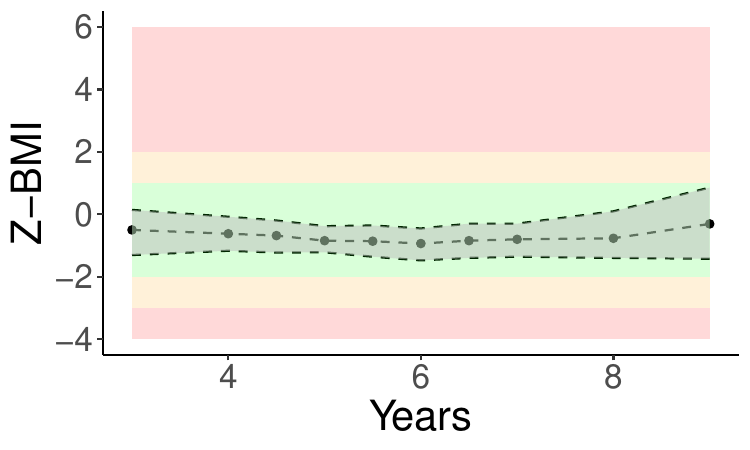}};
			\node[main] (1D) [left= 3 cm of 1A, minimum size = 10cm, fill=blue!20, align=center] {\HUGE{{10\%}} \\ \\ \includegraphics[width=\textwidth]{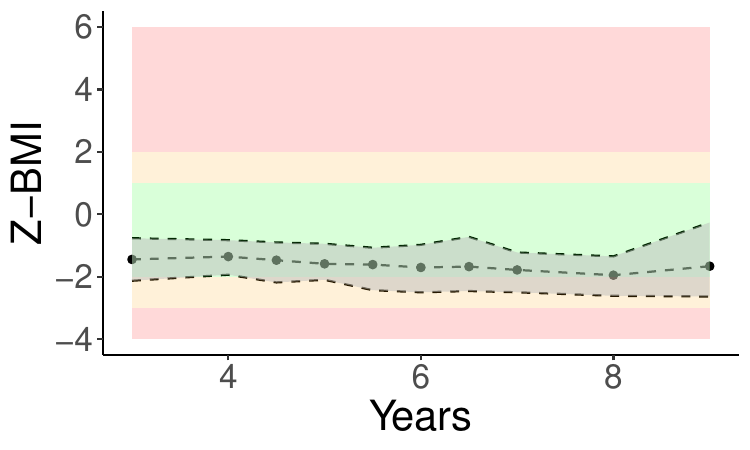}};
			\node[main] (1C) [right= 3 cm of 1B, minimum size = 10cm, fill=blue!20, align=center]{\HUGE{{22\%}} \\ \\ \includegraphics[width=\textwidth]{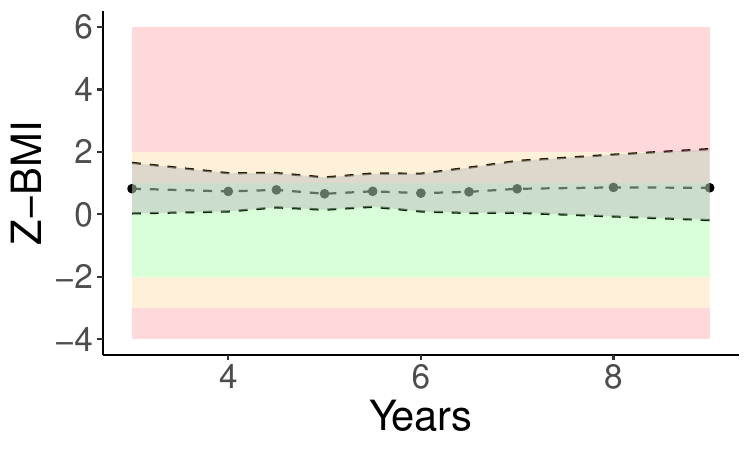}}; 
			\node[main] (1E) [right= 3 cm of 1C, minimum size = 10cm, fill=blue!20, align=center]{\HUGE{{14\%}} \\ \\ \includegraphics[width=\textwidth]{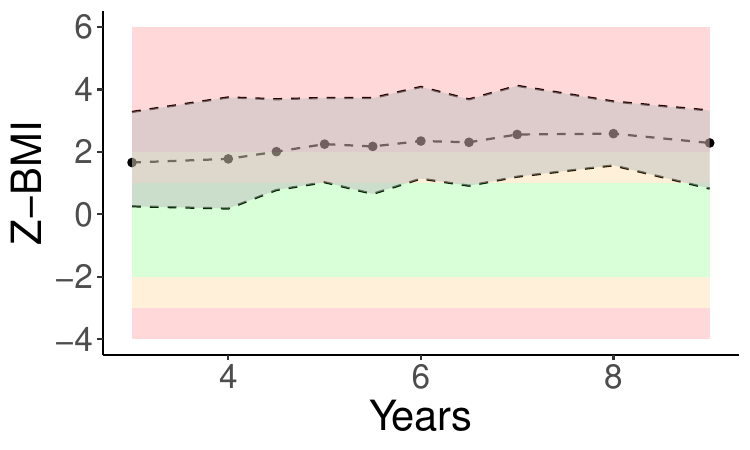}};
			\node[main] (2B) [above= 10 cm of 1B, minimum size = 10cm, fill=red!20, align=center] {\HUGE{{71\%}} \\ \\ \includegraphics[width=.7\textwidth]{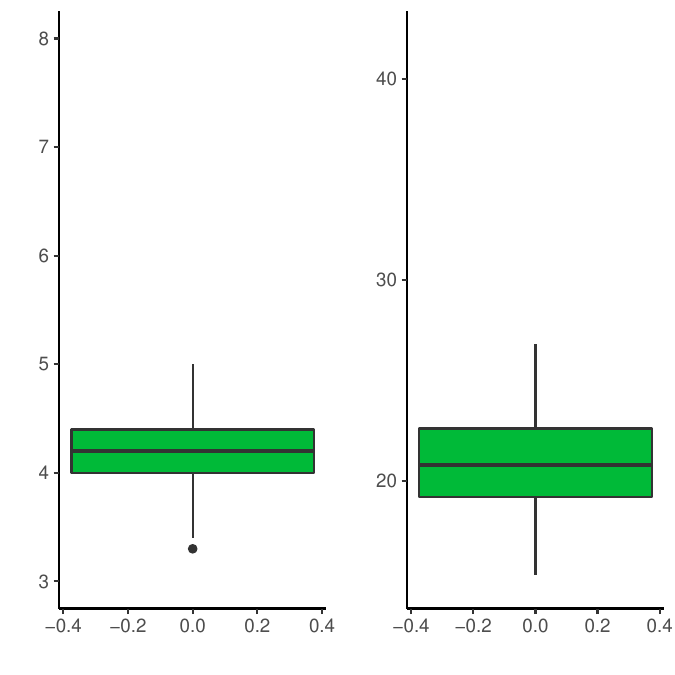}};
			\node[main] (2A) [right= 4 cm of 2B, minimum size = 10cm, fill=red!20, align=center] {\HUGE{{28\%}} \\ \\ \includegraphics[width=.7\textwidth]{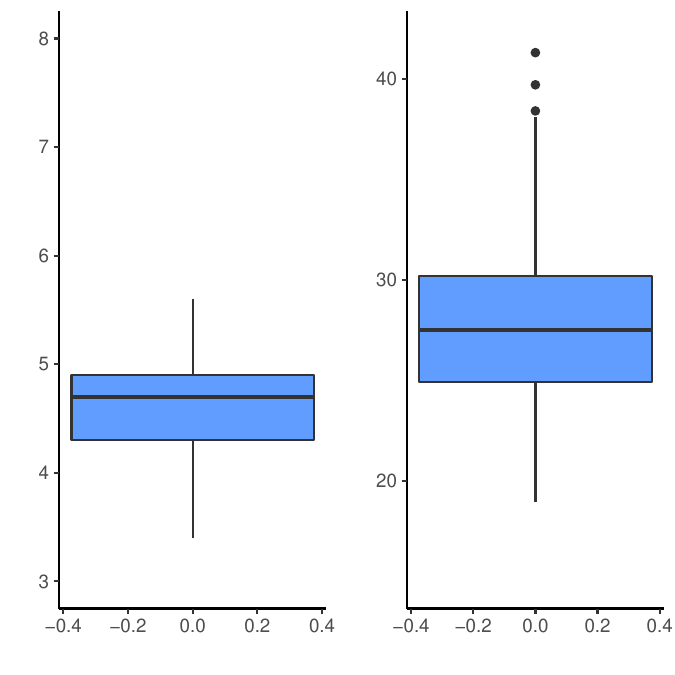}}; 
			\node[main] (2C) [left= 4 cm of 2B, minimum size = 10cm, fill=red!20, align=center] {\HUGE{{1\%}} \\ \\ \includegraphics[width=.7\textwidth]{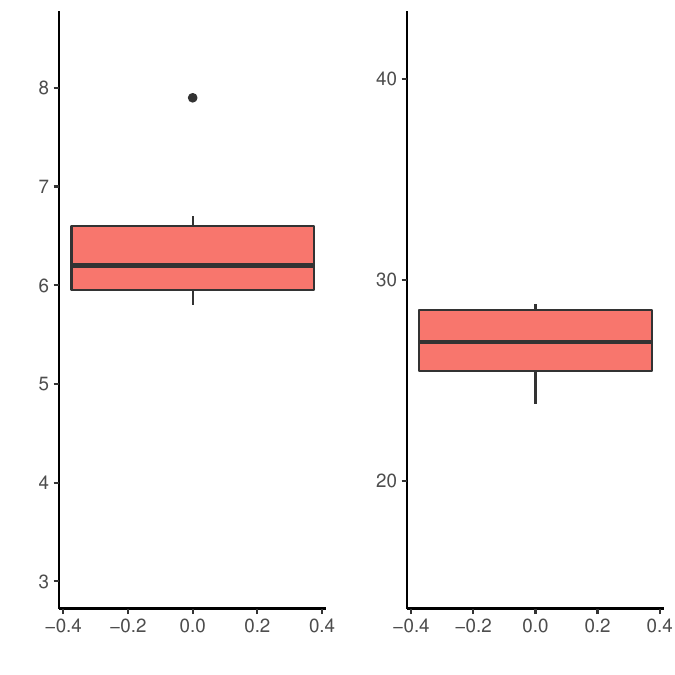}};
			\node[main] (3A) [below left =23cm of 1C, minimum size = 4.6cm, fill=green!20, align=center] {\HUGE{{89\%}} \\ \\ \includegraphics[width=.7\textwidth]{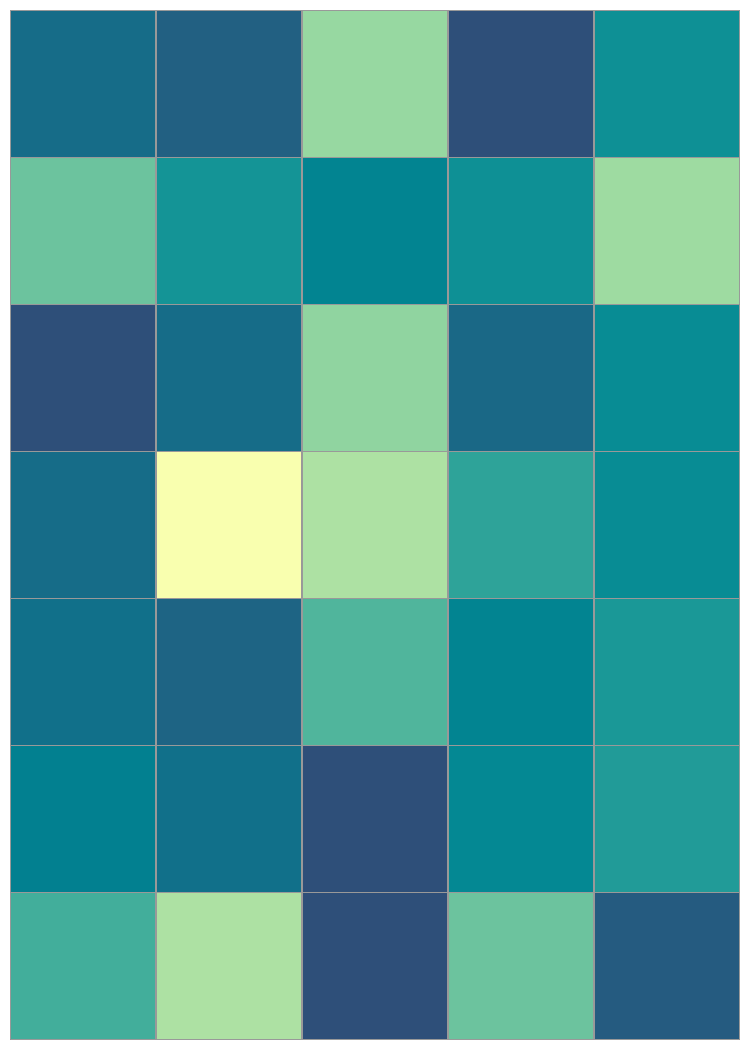}}; 
			\node[main] (3B) [right=of 3A, minimum size = 5.4cm, fill=green!20, align=center] {\HUGE{{10\%}} \\ \\ \includegraphics[width=.7\textwidth]{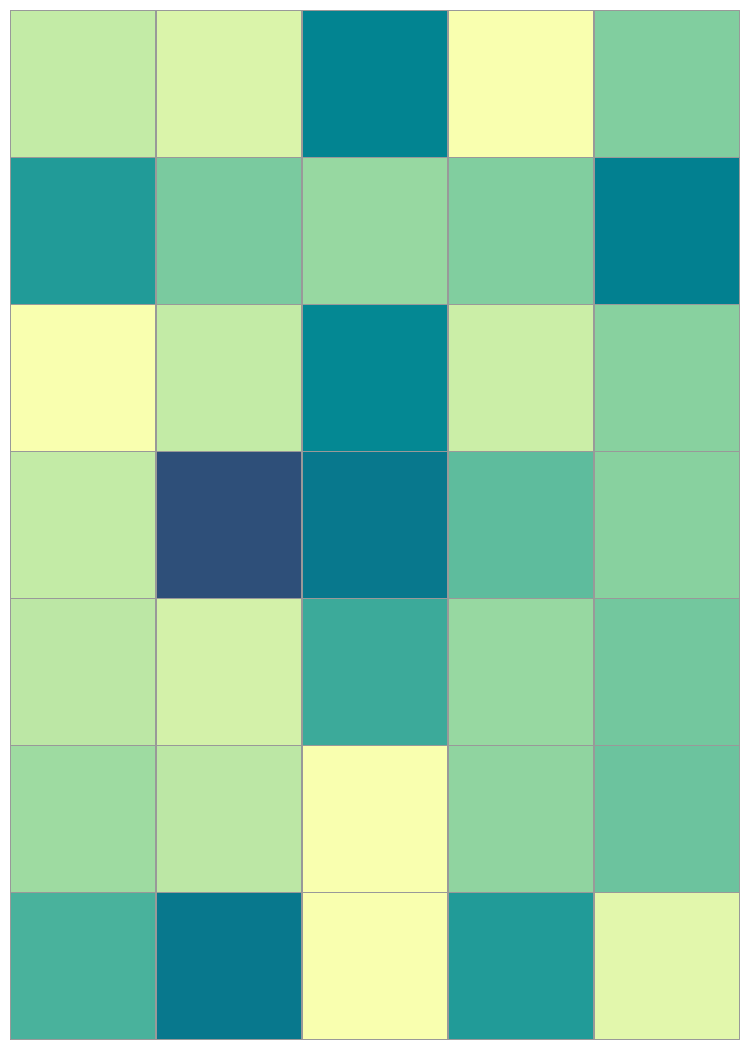}};
			\draw [->, line width = 5.6 mm] (1A) -> (2B) node[pos=0.5,above,font=\HUGE]
			{83\%\textcolor{white}{--------}};
			\draw [->, line width = 5.6 mm] (1B) -> (2B) node[pos=0.5,above,font=\HUGE]
			{74\%\textcolor{white}{--------}};
			\draw [->, line width = 5.6 mm] (1C) -> (2B) node[pos=0.5,above,font=\HUGE]
			{58\%\textcolor{white}{------------}};
			\draw [->, line width = 5.6 mm] (1E) -> (2B) node[pos=0.5,above,font=\HUGE]
			{49\%\textcolor{white}{----------------}};
			\draw [->, line width = 5.6 mm] (1E) -> (2A) node[pos=0.5,above,font=\HUGE]
			{49\%\textcolor{white}{------------}};
			\draw [->, line width = 5.6 mm] (1D) -> (2B) node[pos=0.5,above,font=\HUGE]
			{85\%\textcolor{white}{--------}};
			\draw [->, line width = 5.6 mm] (1A) -> (3A) node[pos=0.5,above,font=\HUGE]
			{90\%\textcolor{white}{----------}};
			\draw [->, line width = 5.6 mm] (1B) -> (3A) node[pos=0.5,above,font=\HUGE]
			{96\%\textcolor{white}{--------}};
			\draw [->, line width = 5.6 mm] (1C) -> (3A) node[pos=0.5,above,font=\HUGE]
			{91\%\textcolor{white}{--------}};
			\draw [->, line width = 5.6 mm] (1D) -> (3A) node[pos=0.5,above,font=\HUGE]
			{98\%\textcolor{white}{-------------}};
			\draw [->, line width = 5.6 mm] (1E) -> (3A) node[pos=0.5,above,font=\HUGE]
			{61\%\textcolor{white}{--------}};
		\end{tikzpicture}
	}
	\caption{\label{fig:graph-metabolites} Estimated clustering configuration for the GUSTO cohort data. Nodes in the graph represent different clusters and colors different layers. The percentage within the nodes denotes the amount of children assigned to that cluster.  Edges are drawn from each growth-trajectory cluster towards the mother cluster and metabolites cluster to which the majority of children in that particular cluster are assigned. }
\end{figure}

In this section, we investigate childhood obesity patterns and their relationship with metabolic pathways, as well as traditional clinical markers for mothers employing the t-HDP.
Data are available on  a sample of $n=553$ children from the  Growing Up in Singapore Towards healthy Outcomes (GUSTO) cohort
study, based in Singapore \citep{soh2014cohort}. 

The first layer of information consists of z-BMI trajectories, including ten unequally spaced measurements per child observed from ages 3 to 9. The second layer contains information on the mother's pre-pregnancy BMI (a known risk factor for childhood obesity) and the fasting oral glucose tolerance test (ogtt) result conducted at week 26 of pregnancy. 

In this third layer, we include concentration data of 35 metabolites measured in the children using NMR spectroscopy. Before applying the t-HDP model, we compute principal components of the metabolite data in the third layer, selecting the first six components based on the scree plot and the elbow method, which collectively explain 66\% of the variability. By clustering on the principal components, we focus on global patterns of the  $35$ metabolites, reducing noise and dimensionality, thus obtaining more robust and interpretable clusters. Data from the same cohort have been also analyzed by \cite{cremaschi2021joint} with the goal of identifying metabolic pathways related to childhood obesity. 

We fit the t-HDP model presented in Section~\ref{ss:tHDP} with multivariate independent Gaussian kernels and Normal-Inverse-Chi-Squared base measures for the vectors of means and variances. We specify a Gamma$(1,1)$ prior on all the concentration parameters. The total number of features is 18, divided into three layers of dimension 10, 2, and 6, respectively. The primary information is the growth trajectory of the child and conditionally on the clustering configuration of the trajectories, we define the model for the mother-layer and the metabolite-layer. We perform 100 000 iterations of the partially collapsed conditional block Gibbs sampler described in Section~S4.2 of the Supplement, discard the first half as burn-in, and apply a thinning of 5 so that the final posterior sample is 10,000 draws. The estimated clustering configurations are summarized in Table~S7.1 in the Supplement and shown in Figure~\ref{fig:graph-metabolites}. A detailed account of the results is provided in Section~S7 of the Supplement. 

The analysis identifies five distinct clusters that represent five different trajectories of z-BMI. The trajectories exhibit relatively stable patterns across the various time points considered but largely vary across clusters in terms of average z-BMI. More precisely, approximately 10\% of children show consistently low z-BMI values (\emph{underweight cluster}), around 14\% of children fell into the cluster characterized by overweight/obesity status (\emph{obesity cluster}), while 26\%, 28\%, and 22\% of children are associated to normal-weight trajectories which are, respectively, below average, equal to average and above average, indicating a healthier weight status as compared to the \emph{underweight cluster} and the \emph{obesity cluster}. 
At layer 2, mothers' clinical profiles are split into three distinct clusters. The first cluster contains a few outliers with exceptionally high glucose levels compared to the average in the sample. The remaining two clusters divide the mothers into two distinct groups. The first group, comprising 71\% of mothers, exhibits \emph{below-average} levels of glucose and BMI. In contrast, the second group, consisting of 28\% of mothers, is characterized by \emph{above-average} levels of both glucose and BMI.
The percentage of children associated with the \emph{below-average cluster} of mothers steadily decreases across the z-BMI clusters as the z-BMI trajectory increases. This finding suggests a positive relationship between the z-BMI trajectories of the child and the clinical markers of the mothers. Specifically, the majority of mothers in the \emph{above-average cluster} have children with an overweight growth trajectory. This association is confirmed in the medical literature \citep[see, for instance,][]{dalrymple2019relationships,josefson2020joint, landon2020relationship, meek2023unwelcome, ormindean2024obesity}. 
At the parallel layer 3, we estimate two distinct clusters characterized by different concentration profiles. The first cluster encompasses approximately 89\% of the children and the second cluster consists of 10\% of the children. 
Furthermore, the results indicate a relationship between obesity and metabolite concentrations. Specifically, conditioning on any of the \emph{normal-weight clusters} or on the \emph{underweight cluster} at layer 1, leads to a very similar distribution of the children across the two metabolite clusters. On the contrary,  conditioning to the \emph{obesity cluster} at layer 1, a drastic variation in the distribution of children across the metabolite clusters is observed. 
These results emphasize the role of metabolite profiles in obesity development, as it is also well documented in the medical literature \citep[see, for instance,][]{perng2014metabolomic,hellmuth2019individual,perng2020metabolomic,handakas2022systematic,schipper2024associations}. The observed associations between obesity trajectories and metabolite clusters provide further evidence of the complex interplay between metabolic factors and weight status. For a detailed account of the metabolite layer results, see Table S7.2 in the Supplement. 

\section{Conclusions and future directions}\label{s:conclusion}

Standard clustering techniques often struggle when applied to datasets collected under a repeated measures design, such as multi-view or longitudinal data. These scenarios require potentially different clustering configurations for each view while still preserving subjects' identities across them. In particular, classical model-based clustering techniques fail to effectively address this issue, as they either impose a single clustering configuration across all views or disregard subjects' identities across views, thereby failing to capture both the multi-view nature of the problem and the repeated measures design underlying the data. \\
To overcome this challenge, we introduce conditional partial exchangeability (CPE), an invariance requirement for the conditional law of the observables in one view, given the clustering configuration of the same units with respect to another view. When satisfied by a probabilistic clustering model, CPE induces dependencies across views while ensuring that subjects' identities are preserved, as formally established in Theorem 1. \\
Furthermore, we introduce, characterize, and apply telescopic clustering models, a novel class of Bayesian mixture models. This class of models highlights that the utility of CPE extends beyond maintaining subject identities in multi-view probabilistic clustering. Rather, its conditional formulation provides a constructive definition that facilitates the development and analysis of diverse clustering processes while ensuring both analytical and posterior computational tractability. We motivate our approach theoretically and conduct extensive comparisons with a range of existing methods, consistently demonstrating that our approach outperforms all competitors. \\
Finally, our framework paves the way for exciting and insightful advancements in the study and development of dependent random partition models. \\
From a theoretical and probabilistic point of view, we have demonstrated that CPE preserves subjects' identities (as formally established in Theorem 1) and shown that some existing Bayesian models preserving subject identities indeed satisfy CPE. However, an interesting open question remains: whether all dependent partition models that preserve subject identities must necessarily satisfy CPE. Establishing this result would allow us to conclude that CPE is not only a sufficient but also a necessary condition for incorporating repeated measures designs into partition models. \\
From a statistical and modeling perspective, further exploration of the t-HDP model (and the telescopic clustering class in general) in its Markovian dependence formulation would be valuable, particularly in identifying conditions for the stationarity of the partition chain's law. Similarly, exploring the properties and applications of a more general polytree-dependent structure, particularly examining the marginal distribution of partitions at the leaves where mutual dependencies can arise, would be highly valuable.

\newpage
\section*{A - Appendix}
\renewcommand{\thesection}{A\arabic{section}}
\setcounter{section}{0}

\section{Proof of Theorem~1}
\begin{proof}[Proof of Theorem~1]
By condition c-i) in the Definition~1 of Section 2 and de Finetti's representation theorem for partial exchangeability \citep{deFinetti1938} there exist $\tilde p_1, \ldots, \tilde p_{K_1}$ random probability measures with distribution $Q$ on $\mathcal{P}^{K_1}$, such that, conditionally on $\rho_1$, for $m$ and $m'$ in $[K_1]$, with  $m\neq m'$, we have
\[
	\mathbb{P}( (X_{2i}, X_{2j})\in A^2 \mid c_{1i} = c_{1j} = m, c_{1k} = m' ) = \int_{\mathcal{P}^{K_1}} \tilde p^2_{m}(A) d Q(\tilde p_1, \ldots, \tilde p_{K_1})
\]
and 
\[
\mathbb{P}( (X_{2i}, X_{2k})\in A^2 \mid c_{1i} = c_{1j} = m, c_{1k} = m' ) = \int_{\mathcal{P}^{K_1}} \tilde p_{m}(A) \tilde p_{m'}(A)d Q(\tilde p_1, \ldots, \tilde p_{K_1})
\]
where, by condition c-ii) in Definition 1, for any $A \in \mathbb{X}_2$, 
$$\mathbb{E}[\tilde p^2_{m}(A)] \geq \mathbb{E}[\tilde p_{m}(A) \tilde p_{m'}(A)].$$
Moreover note that, in general, being $\tilde p_1, \ldots, \tilde p_{K_1}$ dependent,
\[
\mathbb{P}( (X_{2i}, X_{2j})\in A\times B \mid c_{1i}= m, c_{1j} = m' ) = \mathbb{E}[\tilde p_{m}(A)\tilde p_{m'}(B)] \neq \mathbb{E}[\tilde p_{m}(A)] \mathbb{E}[\tilde p_{m'}(B)].
\]
\end{proof}

\section{Proof of Proposition~1}
Before proving Proposition~1, we first introduce the following Lemma.

\begin{lem}
	Given a (non-random) partition $\rho$ of $n$ elements, a vector $\bm{\gamma}=(\gamma_1,\ldots,\gamma_n)$ with binary entries, and a permutation $\sigma:[n]\rightarrow[n]$ of $n$ elements, let 
	\begin{itemize}
		\item $\sigma(\rho)$ be the partition obtained swapping the elements in the sets of $\rho$ accordingly to $\sigma$,
		\item $\mathcal{R}(\bm{\gamma})= \{ i : \gamma_{i} = 1\}$ and $\sigma(\bm{\gamma})=(\gamma_{\sigma(1)},\ldots,\gamma_{\sigma(n)})$
		\item $\rho^{\mathcal{R}(\bm{\gamma})}$ be the ``reduced partition" obtained removing from the sets in $\rho$ all elements that are not in $\mathcal{R}(\bm{\gamma})$ and then removing empty sets.
	\end{itemize}
	then 
	\begin{enumerate}
		\item $\rho^{\mathcal{R}(\bm{\gamma})} = \rho^{\mathcal{R}(\sigma(\bm{\gamma}))} $
		\item $\sigma^{-1}\left(\sigma\left(\rho^{\mathcal{R}^{(\bm{\gamma})}}\right) \right)= \rho^{\mathcal{R}^{(\bm{\gamma})}} $
		\item $\rho^{\mathcal{R}(\bm{\gamma})} = \sigma\left(\rho^{\mathcal{R}(\bm{\gamma})}\right) \text{ for any } \bm{\gamma}\in\{0,1\}^n \quad \text{ iff } \quad\sigma \in \mathcal{P}(n; \rho)$
	\end{enumerate}
	where $\sigma^{-1}$ denotes the inverse of $\sigma$, i.e., $\sigma^{-1}(i) = j$, for $j$ such that $\sigma(j) = i$ and
	$\mathcal{P}(n; \rho)$ denotes the space of permutations of $n$ elements that preserve $\rho$, cf. Definition~1 in Section~2.2. 
\end{lem}
\begin{proof}[Proof of Lemma 1]
	The first statement follows trivially by definition of $\rho^{\mathcal{R}(\bm{\gamma})}$. The second statement follows by the definition of $\sigma^{-1}$ inverse of $\sigma$.
	The last statement follows by considering $\gamma=(1,\ldots,1)$ and the definition of $\mathcal{P}(n; \rho)$. 
\end{proof}

\begin{proof}[Proof of Proposition 1]
	Denoting with $X_{ti}$ a response measured on the $i$th unit at time $t$, for $i=1,\ldots,n$ and $t=1,\ldots,T$, the t-RPM mixture model of \cite{page2022dependent} is defined as 
	\begin{align*}
		X_{ti}\mid \bm{\theta}^{\star}_t, \bm{c}_{t} &\overset{iid}{\sim} k(X_{ti}, \theta^{\star}_{tc_{ti}})\qquad &\text{for } i=1,\ldots,n \text{ and }t=1,\ldots,T\\
		\theta^{\star}_{tj}\mid \mu_t&\overset{ind}{\sim} P_{\mu_t} \qquad &\text{for } j=1,\ldots,K_{t} \text{ and }t=1,\ldots,T\\
		\{\bm{c}_t,\ldots,\bm{c}_T\}\mid \bm{\alpha} &\sim \text{tRPM}(\bm{\alpha}, n)
	\end{align*}
	where $ \bm{\theta}^{\star}_t = (\theta^{\star}_{t1},\ldots, \theta^{\star}_{tK_t})$, $K_t$ is the number of clusters at time $t$, $k$ denotes a kernel, $ P_{\mu_t}$ is an absolutely continuous distribution, $\bm{c}_{t} = (c_{t1},\ldots,c_{tn})$ is the vector of allocation variables encoding the clustering configuration at time $t$, and $\bm{\alpha}=(\alpha_1,\ldots,\alpha_T)\in[0,1]^T$. 
	For the formal and detailed definition of 
	\[
	\{\bm{c}_t,\ldots,\bm{c}_T\}\mid \bm{\alpha} \sim \text{tRPM}(\bm{\alpha}, n)
	\]
	we refer to the paper of \cite{page2022dependent}, even though in the following we describe the core of the construction.
	
	Denoting with $\rho_{t-1}$ the partition encoded by $\bm{c}_{t-1}$, to prove CPE, we need to prove that 
	\[
	p(X_{t1},\ldots,X_{tn}\mid \rho_{t-1}) =
	p(X_{t\sigma(1)},\ldots,X_{t\sigma(n)}\mid \rho_{t-1})
	\]
	for any $\sigma \in \mathcal{P}(n; \rho_{t-1})$, where, we recall, that
	$\mathcal{P}(n; \rho_{t-1})$ denotes the space of permutations of $n$ elements that preserve $\rho_{t-1}$, see Section~2. 
	
	Given a partition $\rho$, we denote with $\sigma(\rho)$ the partition obtained by swapping the elements in the sets of $\rho$ accordingly to the permutation $\sigma$.
	In the t-RPM mixture, the conditional law of $(X_{t\,i})_{i=1}^n$ conditionally on the partition at the previous time point $\rho_{t-1}$, is defined such that
	\[
	p(X_{t1},\ldots,X_{tn}\mid \rho_{t-1})=\sum_{\lambda}p(X_{t1},\ldots,X_{tn}\mid \rho_{t} = \lambda)\enskip\mathbb{P}(\rho_{t} = \lambda\mid \rho_{t-1})
	\]
	where, the sum runs over all partitions $\lambda$ of $n$ elements. Each summand in the sum above is given by the product of two factors. For the first factor, we have trivially that:
	\[
	p(X_{t1},\ldots,X_{tn}\mid \rho_{t} = \lambda) = p(X_{t\sigma(1)},\ldots,X_{t\sigma(n)}\mid \rho_t = \sigma(\lambda))
	\]
	for any permutation $\sigma$ of $n$ elements.
	For what concerns the second factor, the conditional distribution $\mathbb{P}[\rho_{t} = \lambda\mid \rho_{t-1}]$ is defined by the introduction of the binary latent variables in $\gamma_{t} = (\gamma_{1t},\ldots,\gamma_{nt})$. The latent variables identify which subjects at time $t-1$
	will be considered for possible cluster reallocation at time $t$. Specifically, let $\gamma_{it}$
	be defined as
	\[
	\gamma_{it} =
	\begin{cases}
		1&\text{if unit $i$ is not reallocated when moving from time } t-1 \text{ to } t\\
		0& \text{otherwise}
	\end{cases}
	\]
	so that
	\[
	\mathbb{P}[\rho_{t} = \lambda\mid \rho_{t-1}] = \sum_{\gamma_{t}} \mathbb{P}[\rho_{t} = \lambda\mid \gamma_{t}, \rho_{t-1}]\enskip p(\gamma_{t})
	\]
	where the sum runs over all binary vectors of length $n$ and $p(\gamma_{t}) = \alpha_t^{\sum_{i=1}^n \gamma_{ti}}$. 
	Each summand in the sum above is given by the product of two factors. The second factor $p(\gamma_{t})$ is invariant with respect to any permutation $\sigma$ of $n$ elements. 
	
	Thus, denoting with $\sigma(\gamma_{t})$ the vector $(\gamma_{\sigma(1)t},\ldots,\gamma_{\sigma(n)t})$, for any permutation $\sigma$ of $n$ elements, to prove that $(X_{t\,i})_{i\geq1}$ is conditionally partially exchangeable with respect to $\rho_{t-1}$, we need to prove that 
	\[
	\mathbb{P}[\rho_{t} = \lambda\mid \gamma_{t}, \rho_{t-1}] = \mathbb{P}[\rho_t = \sigma(\lambda)\mid \sigma(\gamma_{t}), \rho_{t-1}]
	\]
	for any $\sigma \in \mathcal{P}(n; \rho_{t-1})$. 
	
	In t-RPM, the left and right hand side of the equation above are respectively
	\[
	\mathbb{P}[\rho_{t} = \lambda\mid \gamma_{t}, \rho_{t-1}] = \frac{\mathbb{P}[\rho_t = \lambda] \mathbbm{I}(\lambda \in P(\gamma_t, \rho_{t-1}))}{\sum_{\lambda'}\mathbb{P}[\rho_t = \lambda'] \mathbbm{I}(\lambda' \in P(\gamma_t, \rho_{t-1}))}
	\]
	and
	\[
	\mathbb{P}[\rho_t = \sigma(\lambda)\mid \sigma(\gamma_{t}), \rho_{t-1}] = \frac{\mathbb{P}[\rho_t = \sigma(\lambda)] \mathbbm{I}(\sigma(\lambda) \in P(\sigma(\gamma_t), \rho_{t-1}))}{\sum_{\lambda'}\mathbb{P}[\rho_t = \sigma(\lambda')] \mathbbm{I}(\sigma(\lambda') \in P(\sigma(\gamma_t), \rho_{t-1}))}
	\]
	where, the sums at the denominators runs over all partitions $\lambda'$ of $n$ elements, $\mathbb{I}$ is the indicator function, and $P(\gamma_t, \rho_{t-1})$ denotes the collection of partitions at time $t$ that
	are compatible with $\rho_{t-1}$ based on 
	$\gamma_{t}$. This collection is the one denoted by $P_{C_t}$ in the paper of \cite{page2022dependent}. 
	
	By marginal exchangeability of $\rho_t$, we have that for any $\sigma$
	\[
	\mathbb{P}[\rho_t = \lambda] = \mathbb{P}[\rho_t = \sigma(\lambda)]
	\]
	
	Consider now the indication functions
	$\mathbbm{I}(\lambda \in P(\gamma_t, \rho_{t-1}))$ and let $\mathcal{R}_t = \{ i : \gamma_{it} = 1\}$ be the sets of indices of those subjects which will not be considered for reallocation time $t$. \cite{page2022dependent} show that 
	\[
	\mathbbm{I}(\lambda \in P(\gamma_t, \rho_{t-1}))=\begin{cases}
		1&\lambda^{\mathcal{R}_t}=\rho_{t-1}^{\mathcal{R}_t}\\
		0&\text{otherwise}\\
	\end{cases}
	\]
	and, thus 
	\[
	\mathbbm{I}(\sigma(\lambda) \in P(\sigma(\gamma_t), \rho_{t-1}))=\begin{cases}
		1&\sigma(\lambda)^{\sigma(\mathcal{R}_t)}=\rho_{t-1}^{\sigma(\mathcal{R}_t)}\\
		0&\text{otherwise}\\
	\end{cases}
	\]
	where $\rho^{\mathcal{R}_t}$ is the \emph{reduced} partition obtained removing from the sets in $\rho$ all elements that are not in the set $\mathcal{R}_t$. 
	
	By Lemma 1, we have 
	\[
	\sigma(\lambda)^{\sigma(\mathcal{R}_t)}=\rho_{t-1}^{\sigma(\mathcal{R}_t)}\quad \text{ iff } \quad\sigma(\lambda)^{\mathcal{R}_t}=\rho_{t-1}^{\mathcal{R}_t} \quad \text{ iff } \quad \lambda^{\mathcal{R}_t}=\sigma^{-1}\left(\rho_{t-1}^{\mathcal{R}_t}\right)
	\]
	Therefore, by Lemma 1, $\mathbbm{I}(\lambda \in P(\gamma_t, \rho_{t-1})) = \mathbbm{I}(\sigma(\lambda) \in P(\sigma(\gamma_t), \rho_{t-1}))$ for any possible realization of $\gamma_t$ if and only if 
	\[
	\rho_{t-1}^{\mathcal{R}_t}=\sigma^{-1}\left(\rho_{t-1}^{\mathcal{R}_t}\right) \quad \text{ iff } \quad 
	\sigma\left(\rho_{t-1}^{\mathcal{R}_t}\right) = \rho_{t-1}^{\mathcal{R}_t} \quad \text{ iff } \quad \sigma \in \mathcal{P}(n; \rho)
	\]
	which proves that t-RPM mixtures are conditionally partially exchangeable. 
	
	To prove that t-RPM mixture are not conditionally exchangeable, consider the counterexample with $n=3$, $\rho_{t-1} = \{\{1,2\},\{3\}\}$, $\rho_t =\{\{1\},\{2,3\}\}$, and $\sigma=(1,3)$. In such a case, the permutation $\sigma$ does not preserve $\rho_{t-1}$ and, as a result, the law of $\rho_t$ conditionally of $\rho_{t-1}=\{\{1,2\},\{3\}\}$ differs from the law of $\rho_t$ conditionally of $\rho_{t-1}=\{\{1\},\{2,3\}\}$. As a result, the conditional law of the sequence of observations at time $t$ is not invariant to any permutation, as prescribed by conditional exchangeability. 
\end{proof}

\section{Proof of Proposition~2}
Before proving Proposition~2, we first introduce the following Lemma.
\begin{lem}
	Given a partitions $\rho$ of $n$ elements and a permutation $\sigma$, 
	\[
	\sigma \in  \mathcal{P}(n; \rho) \qquad \text{ iff }\qquad 
	\sigma^{-1} \in  \mathcal{P}(n; \rho)
	\]
	where $\sigma^{-1}$ denotes the inverse of $\sigma$, i.e., $\sigma^{-1}(i) = j$, for $j$ such that $\sigma(j) = i$ and
	$\mathcal{P}(n; \rho)$ denotes the space of permutations of $n$ elements that preserve $\rho$, cf. Section~2. 
\end{lem}
\begin{proof}[Proof of Proposition~2]
	If $(X_{1i},\ldots, X_{Ji})_{i\geq1}$ follows the separate exchangeable random partition mixture of \cite{lin2021separate}, then 
	\begin{align*}
		X_{ji}\mid S_j = k, M_{ik} = \ell &\overset{ind}{\sim} k(X_{ji}, \theta^{\star}_{\ell})\qquad \text{for } i=1,2,\ldots \text{ and }j=1,\ldots,J\\
		\mathbb{P}(M_{ik} = \ell \mid w_{k\ell}) = w_{k\ell}\quad&\quad \bm{w}_k = (w_{k1},w_{k2},\ldots) \overset{iid}{\sim} \text{GEM}(\alpha)\\
		\mathbb{P}(S_{j} = k \mid \pi_{k}) = \pi_{k}\quad&\quad \bm{\pi}= (\pi_{1},\pi_{2},\ldots) \sim \text{GEM}(\beta)\\
		\theta^{\star}_{\ell}\overset{iid}{\sim}G_0
	\end{align*}
	where $\text{GEM}(\alpha)$
	denote a stick-breaking prior for a sequence of weights \citep{sethuraman1994constructive} and $G_0$ is
	an absolutely continuous distribution. The partition $\rho_j$ corresponding to the $j$th layer $(X_{ji})_{i\geq1}$ is encoded by $(M_{iS_j})_{i\geq1}$ and for any $n\geq 1$, $j,j' \in [J]$ and any realization $\rho$ of the partition $\rho_{j'}$, we have
	\begin{align*}
		p(X_{j1},\ldots,X_{jn}\mid \rho_{j'}=\rho) 
		=&\, \mathbb{P}[S_j = S_{j'}] \,p(X_{j1},\ldots,X_{jn}\mid \rho_j = \rho)+
		\mathbb{P}[S_j \neq S_{j'}] p(X_{j1},\ldots,X_{jn})
	\end{align*}
	and, similarly, 
	\begin{align*}
		p(X_{j\sigma(1)},\ldots,X_{j\sigma(n)}\mid \rho_{j'}=\rho) =&\, \mathbb{P}[S_j = S_{j'}] \,p(X_{j\sigma(1)},\ldots,X_{j\sigma(n)}\mid \rho_j = \rho)\\
		&\,+\mathbb{P}[S_j \neq S_{j'}] \, p(X_{j\sigma(1)},\ldots,X_{j\sigma(n)}).
	\end{align*}
	Thus,
	\begin{align*}
		D:=&\,p(X_{j1},\ldots,X_{jn}\mid \rho_{j'}=\rho) - p(X_{j\sigma(1)},\ldots,X_{j\sigma(n)}\mid \rho_{j'}=\rho) \\
		=&\,  \mathbb{P}[S_j = S_{j'}] \left(\,p(X_{j1},\ldots,X_{jn}\mid \rho_j = \rho) - p(X_{j\sigma(1)},\ldots,X_{j\sigma(n)}\mid \rho_j = \rho)\,\right)\\
		=&\,  \mathbb{P}[S_j = S_{j'}] \left(\, p(X_{j1},\ldots,X_{jn}\mid \rho_j = \rho) - p(X_{j1},\ldots,X_{jn}\mid \rho_j = \sigma^{-1}(\rho))\right)
	\end{align*}
	By Lemma~2, for any $\sigma \in \mathcal{P}(n; \rho)$, we have $D = 0$, where, we recall, that
	$\mathcal{P}(n; \rho)$ denotes the space of permutations of $n$ elements that preserve $\rho$, see Definition~1 in Section~2. 
	
	To prove that the separate exchangeable random partition mixture is not conditionally exchangeable, consider the counterexample with $n=3$, $\rho_{j'} = \{\{1,2\},\{3\}\}$, $\sigma=(1,3)$ and $(X_{j1},X_{j2},X_{j3}) \in (d(\bar\theta^{\star}_{\ell}-\epsilon) ,d(\bar\theta^{\star}_{\ell}+\epsilon),d(\bar\theta^{\star}_{\ell}+2\epsilon))$, where $\bar\theta^{\star}_{\ell} = \mathbb{E}[\theta^{\star}_{\ell}]$, $\epsilon>0$ and $dy = [y, y + \nu)$, with $\nu$ arbitrarily small. In such a case, the permutation $\sigma$ does not preserve $\rho_{j'}$ and, as a result, the law of $\rho_j$ conditionally of $\rho_{j'}=\{\{1,2\},\{3\}\}$ differs from the law of $\rho_j$ conditionally of $\rho_{j'}=\{\{1\},\{2,3\}\}$. The conditional law of the sequence of observations corresponding to $\rho_j$ is not invariant to any permutation, as prescribed by conditional exchangeability.  
\end{proof}

\section{Proof of Proposition~3}
\begin{proof}[Proof of Proposition~3]
	Denoting with $X_{ix}$ the response measured on the ith unit corresponding to covariate's value $x \in \mathcal{X}$ and following a mixture model with mixing probability provided by the dependent processes in \cite{maceachern2000dependent}, then 
	\begin{align*}
		X_{ix}\mid \bm{\theta}^{\star}_x &\overset{iid}{\sim} k(X_{ix}, \theta_{xi})\qquad \text{for } i=1,\ldots,n \text{ and for any } x \\
		\theta_{xi}&\overset{ind}{\sim} G_{x} \\
		\{G_{x}: x \in \mathcal{X}\}&\sim DDP
	\end{align*}
	For a formal and detailed definition of 
	$\{G_{x}: x \in \mathcal{X}\}\sim DDP$
	we refer to the recent review paper of \cite{quintana2022dependent}.
	
	Denoting with $\rho_{x}$ the partition induced by $G_{x}$, for any $\sigma$ permutation of $n$ elements, we have
	\begin{align*}
		&\,p(X_{x'1},\ldots,X_{x'n}\mid \rho_{x}) =
		\int p(X_{x'1},\ldots,X_{x'n}\mid G_{x'},\rho_{x})\,\text{d}\,p(G_{x'}\mid \rho_{x})\\
		&\,=\int p(X_{x'1},\ldots,X_{x'n}\mid G_{x'})\text{d}\,p(G_{x'}\mid \rho_{x})=\int p(X_{x'\sigma(1)},\ldots,X_{x'\sigma(n)}\mid G_{x'})\text{d}\,p(G_{x'}\mid \rho_{x})\\
		&\,= p(X_{x'\sigma(1)},\ldots,X_{x'\sigma(n)}\mid \rho_{x}).
	\end{align*}
\end{proof}

\section{Proof of Theorem~2}
\begin{proof}[Proof of Theorem 2]
	Note that, for any $n\geq1$,
	the second layer observations, admit the following almost sure representation in terms of a latent collection of probability measures $(\tilde q_1,\ldots, \tilde q_n)$ such as 
	\begin{align*}
		X_{2i}\mid \tilde q_i \overset{ind}{\sim} \int k_2(X_{2i};\xi) \tilde q_i(\dd \xi),\qquad \text{where} \enskip \tilde q_i \mid \bm{w},\tilde p_{21},\ldots, \tilde p_{2M}\overset{iid}{\sim} \sum_{m=1}^M w_m \delta_{ \tilde p_{2m}}
	\end{align*}
	and $\bm{w}$ is the sequence of weights in the almost-sure representation of $\tilde p_{1}$.
	Moreover, conditioning both layers to the allocations variables $\bm{c}_1$ and the unique values $\bm{\theta^{\star}}$ corresponding to the first layer, we get 
	\begin{align*}
		(X_{1i},X_{2i})\mid c_{1i} = m,\, \theta^{\star}_m,\, \tilde p_{21},\ldots, \tilde p_{2M} \overset{ind}{\sim} k_1(X_{1i};\theta^{\star}_m)\left(\sum_{s=1}^S q_{ms}  k_2(X_{2i};\xi^{\star}_s)\right).
	\end{align*}
\end{proof}

\section{Proof of Proposition~4}
\begin{proof}[Proof of Proposition~4]
	Note that, for any $i\neq j$, by exchangeability of the rows in the data matrix, we have 
	\begin{equation*}
		\mathbb{P}(c_{\ell i}=c_{\ell j}) =
		\mathbb{P}(c_{\ell 1}=c_{\ell 2})
		\qquad \text{and} \qquad \mathbb{P}(c_{\ell i}=c_{\ell j}, c_{\ell' i}=c_{\ell' j}) =
		\mathbb{P}(c_{\ell 1}=c_{\ell 2} , c_{\ell' 1}=c_{\ell' 2})
	\end{equation*}
	Thus
	\[
	\text{$\tau$} = \frac{\mathbb{P}[c_{21} = c_{22}  \mid c_{11} = c_{12}] - \mathbb{P}[c_{21} = c_{22} \mid c_{11} \neq c_{12}]}{\mathbb{P}[c_{21} = c_{22}  \mid c_{11} = c_{12}]}
	\]
	where the event $c_{\ell 1} = c_{\ell 2}$ coincides with the event $K_{\ell 2} = 1$ and
	$c_{\ell 1} \neq c_{\ell 2}$ with the event $K_{\ell 2} = 2$,
	where $K_{\ell n}$ denote the number of cluster at layer $\ell$ in a sample of $n$ subjects. 
	Similarly,
	\begin{align*}
		ER=&\, {n\choose 2} ^{-1}\mathbb{E}\left[\sum_{i=1}^n\sum_{j=i+1}^n \mathbbm{1}(c_{1i}=c_{1j})\mathbbm{1}(c_{2i}=c_{2j}) +  \sum_{i=1}^n\sum_{j=i+1}^n \mathbbm{1}(c_{1i}\neq c_{1j})\mathbbm{1}(c_{2i} \neq c_{2j})\right] \\
		=&\, \mathbb{P}(c_{1i}=c_{1j},c_{2i}=c_{2j}) + \mathbb{P}(c_{1i}\neq c_{1j},c_{2i} \neq c_{2j}) \\
		=&\mathbb{P}(c_{11}=c_2,s_1=s_2) + \mathbb{P}(c_{11}\neq c_2,s_1 \neq s_2) \\
		=&\,\mathbb{P}(K_{12} = 1, K_{22} = 1) + \mathbb{P}(K_{12} = 2, K_{22} = 2).
	\end{align*}
\end{proof}
\section{Proof of Theorem~3 and Corollary~1} 
Proof of Theorem~3 follows directly by combining equations (6) and (7) in Section 4. Corollary~1 follows directly from Theorem~3 and Proposition~4.
\section{Proof of Theorem~4 and Corollary~2} 
\begin{proof}[Proof of Theorem~4]
The marginal EPPF of the partition at layer 1  is a well-known result \citep[see, e.g.,][]{green2001modelling,mccullagh2008many,miller2018mixture, argiento2019infinity}.  Considering a specific partition $\rho_1$ into $K_{1n}$ sets of the $n$ observations, under eq. (8) in Section 5, we have that 
\begin{equation*}
	%\label{eq:firstEPPF}
	p(\rho_1) = V(n,K_{1n}) \, \prod_{m = 1}^{K_{1n}} \frac{\Gamma(\gamma + n_{m})}{\Gamma(\gamma)},
\end{equation*}
where  $n_{m}$ is the frequency of the $m$th cluster in order of appearance, i.e.,
\[
n_{m} = \sum_{i=1}^n \mathbbm{1}_{m}(c^{\star}_{1i}) \text{ with } \sum_{m=1}^{K_{1n}} n_{m} = n \qquad \text{and} \qquad V(n,K_{1n}) = \sum\limits_{M=1}^{+\infty} \frac{M_{(K_{1n})}}{(\gamma K_{1n})^{(n)}} p_M(M)
\]
where $x^{(k)} = \Gamma(x+k)/\Gamma(x) = x(x+1)\ldots(x+k-1)$ and $x_{(k)} = \Gamma(x+1) / \Gamma(x-k+1) = x(x-1)\ldots(x-k+1)$, where $\Gamma(x)$ denote the Gamma function in $x$ and $x_{(0)}=1$ and $x_{(0)}=1$ by convention. 
While from equation (9), we have that 
\[
p(\rho_2 \mid \rho_1) = (1- \omega) \mathbbm{1}(\rho_1=\rho_2) + \omega \,V(n,K_{2n}) \, \prod_{s = 1}^{K_{2n}} \frac{\Gamma(\alpha + \sum_{m=1}^{K_{1n}}n_{ms})}{\Gamma(\alpha)}
\]
where $n_{ms}$ is the number of observations in the first-layer cluster $m$ and second-layer cluster $s$, when the clusters are in order of appearance. 

Proof of Theorem~4 follows directly by combining the two partition functions above. 
Corollary~2 follows directly from Theorem~4 and Proposition~4.
\end{proof}
\bibliographystyle{agsm}

\bibliography{references.bib} 
\end{document}